%% file: main.tex
\documentclass{llncs}

\usepackage[utf8]{inputenc}

\usepackage{etoolbox}
\usepackage{cite}
\usepackage{soul}
\usepackage{graphicx}
\usepackage{pgf}
\usepackage{tikz}
\usepackage{amssymb}
\usepackage{amsmath}
\usepackage{bbm}
\usetikzlibrary{arrows,automata,shapes,decorations,calc}
\usetikzlibrary{decorations.pathmorphing}
\usetikzlibrary{decorations.markings}
\usetikzlibrary{decorations.pathreplacing}
\usetikzlibrary{patterns}

\usepackage{verbatim} 
\usepackage{wrapfig}
\usepackage{color}
\usepackage{url}
\usepackage{ifthen}

\newtheorem{mydef}{\bfseries Definition}
\newtheorem{mytheo}{\bfseries Theorem}
\newtheorem{mylemma}{\bfseries Lemma}

\newtheorem{myexample}{\bfseries Example}
\newtheorem{myprop}{\bfseries Proposition}

\newtheorem{myassumption}{\bfseries Assumption}

\newcommand{\mypara}[1]{\medskip\noindent {\bf #1}\quad}

\newcommand{\theoremlike}[2]{\par\medskip\penalty-250\refstepcounter{mylemma}{{\bfseries\noindent#2
			\ref{#1}.}}}
\newcommand{\thmhelperpre}[2]{\theoremlike{#1}{#2}}
\newcommand{\thmhelperpost}{\par\medskip}

\newenvironment{reftheo}[1]{\thmhelperpre{#1}{Theorem}}{\thmhelperpost}
\newenvironment{refprop}[1]{\thmhelperpre{#1}{Proposition}}{\thmhelperpost}

\newcommand{\mal}{\textit{Mal}}
\newcommand{\app}{\textit{App}}
\newcommand{\att}{\textit{Att}}

\newcommand{\login}{\textit{login}}
\newcommand{\logout}{\textit{logout}}
\newcommand{\lookup}{\textit{lookup}}

\newcommand{\img}{\mathcal{G}} %M
\newcommand{\pomdp}{\mathcal{P}} %M
\newcommand{\obs}{\mathcal{O}}
\newcommand{\obsfunc}{O}
\newcommand{\transfunc}{P}
\newcommand{\syncmod}{\mathcal{S}}

\newcommand{\agt}{\textit{Plr}}
\newcommand{\act}{\textit{Act}}
\newcommand{\lab}{\textit{Label}}

\newcommand{\last}{\textup{last}}

\makeatletter
\newsavebox{\@brx}
\newcommand{\llangle}[1][]{\savebox{\@brx}{\(\m@th{#1\langle}\)}%
  \mathopen{\copy\@brx\kern-0.5\wd\@brx\usebox{\@brx}}}
\newcommand{\rrangle}[1][]{\save	
box{\@brx}{\(\m@th{#1\rangle}\)}%
  \mathclose{\copy\@brx\kern-0.5\wd\@brx\usebox{\@brx}}}
\makeatother

\newcommand{\syn}{\textup{Sync}}

\newcommand{\play}{\textit{Play}}
\newcommand{\hist}{\textit{Hist}}
\newcommand{\choice}{\mathcal{C}}

\newcommand{\enab}{\textup{En}}

\newcommand{\snote}[1]{\textcolor{red}{\textbf{#1}}}
\newcommand{\hnote}[1]{\textcolor{blue}{\textbf{#1}}}
\newcommand{\jnote}[1]{\textcolor{green}{\textbf{#1}}}

\newcommand{\marko}[1]{
	\ifthenelse{\equal{#1}{}}{\mathop{\rightsquigarrow}}{\overset{#1}{\rightsquigarrow}}
	}
\newcommand{\inter}[1]{
	\ifthenelse{\equal{#1}{}}{\mathop{\hookrightarrow}}{\overset{#1}{\hookrightarrow}}
	}
	
\newcommand{\intersmall}[1]{
	\mathop {\smash\hookrightarrow}\limits^{\vrule width 0pt height 0pt depth 4pt\smash{#1}}	
}
\newcommand{\markosmall}[1]{
	\mathop {\smash\rightsquigarrow}\limits^{\vrule width 0pt height 0pt depth 4pt\smash{#1}}	
}

\newcommand{\arrow}[2]{\mathrel{\raisebox{-2pt}{$\xrightarrow{#1}_{#2}$}}}

\newcommand{\pr}{\textup{Pr}}

\newcommand{\cyl}{\textup{Cyl}}
\newcommand{\cylta}{\textup{Cyl}^\textup{ia}}

\newcommand{\exptime}{\textsc{ExpTime}}

\newcommand{\vect}[1]{\boldsymbol{#1}}
\newcommand{\dist}{\Delta}

\newcommand{\Nset}{\mathbb{N}}
\newcommand{\Nseto}{\Nset_0}

\newcommand{\Qset}{\mathbb{Q}}
\newcommand{\Rset}{\mathbb{R}}

\newcommand{\Rsetpo}{\mathbb{R}_{\ge 0}}

\newcommand{\sigmafielda}{\mathcal{H}}
\newcommand{\sigmafieldb}{\mathcal{I}}

\newcommand{\mycut}[1]{}

\newcommand{\initstate}{s^{in}}

\newcommand{\sigmafield}{\mathcal{F}}

\newcommand{\mdp}{\mathcal{M}}
\newcommand{\tra}{P}
\newcommand{\states}{S}
\newcommand{\actions}{A}

\newcommand{\structure}{\mathcal{G}}
\newcommand{\donothing}{\varnothing}

\newcommand{\row}{\mathrm{row}}
\newcommand{\col}{\mathrm{col}}

\newcommand{\sync}{\mathrm{sync}}

\usepackage[belowskip=-5pt,aboveskip=0pt]{caption}

\begin{document}

\title{Distributed Synthesis in Continuous Time}

\author{Holger Hermanns\inst{1} \and Jan Kr\v{c}\'al\inst{1} \and Steen Vester\inst{2}}

\institute{Saarland University -- Computer Science, Saarbr\"{u}cken, Germany\\ \email{\{hermanns, krcal\}@cs.uni-saarland.de}
\and Technical University of Denmark\\ \email{stve@dtu.dk}}
%\institute[2]{\\
%	\texttt{stve@dtu.dk}}
%
%	\texttt{\{hermanns, krcal\}@cs.uni-saarland.de}
\maketitle

\begin{abstract}
We introduce a formalism modelling communication of distributed agents strictly in \emph{continuous-time}. Within this framework, we study the problem of synthesising local strategies for individual agents such that a specified
set of goal states is reached, or reached with at least a given probability. 
The flow of time is modelled explicitly based on continuous-time randomness, with two natural implications: 
First, the non-determinism stemming from interleaving disappears.
Second, when we restrict to a subclass of \emph{non-urgent} models, the quantitative value problem for two players can be solved in EXPTIME.
Indeed, the explicit continuous time enables players to communicate their states by delaying synchronisation (which is unrestricted for non-urgent models).
In general, 
the problems are undecidable already for two players in the quantitative case and three players in the qualitative case.
The qualitative undecidability is shown by a reduction to decentralized POMDPs for which we provide the strongest (and rather surprising) undecidability result so far.
%

%We introduce a formalism modelling communication of distributed agents strictly in continuous-time. Within this framework, we study the problem of synthesising local strategies for individual agents such that a specified
%set of goal states is reached, or reached with at least a given probability. 
%The flow of time is modelled explicitly based on continuous-time randomness, with two natural implications: First, the non-determinism stemming from interleaving disappears. Second, when we restrict to a subclass of \emph{non-urgent} models, the reachability problems for two players can be solved in EXPTIME. Indeed, the explicit continuous time enables players to communicate their states by delaying synchronisation (which is unrestricted for non-urgent models). In general, the problems are undecidable already for two players in the quantitative case and three players in the qualitative case. The qualitative undecidability is shown by a reduction to decentralized POMDPs for which we provide the strongest (and rather surprising) undecidability result so far.

\end{abstract}

\input{intro}
\input{defs}

\input{results-schedulers}

\input{results-reductions}

\input{results-reachability}

\input{discussion}

\newpage

\bibliographystyle{plain}
\bibliography{biblio}

\newpage
\appendix

\input{appendix.tex}

\end{document}

%% file: intro.tex
\section{Introduction}\label{sec:intro}

\emph{Distributed} self-organising and self-maintaining systems are
posing interesting design challenges, and have been subject to many
practical~\cite{sinopoli} as well as theoretical~\cite{pnueli-rosner,pnueli-rosner-async}
investigations.  Distributed systems interact in real time,
and one very general way to reason about their timing behaviour is to
assume that arbitrary continuous probability distributions govern the
timing of local steps as well as of communication steps. We are
interested in how foundational properties of such distributed
systems differ from models where timing is abstracted. %\jnote{cannot we motivate more than that they natually interact in continuous/real time?}
%\hnote{Don't know how to do that better. We state it as a fact above.} 
As principal means of communication we consider symmetric handshake
communication, since it can embed other forms of communication
faithfully~\cite{milner-calculi-for-synchrony-and-asynchrony,baier-katoen}
including asynchronous and input/output-separated communication.

% As an example, consider the problem of organising the communication between a handful of terrorists. Each of them has a prescribed behaviour in terms of states, local transitions, (labelled) synchronising transitions, and delay transitions, where the delays are given by arbitrary continuous probability distribution over the real time. 
 
% A (handshake) synchronisation is assumed to take place if all components able to do so agree on the same label. 
%The qualitative question is now how to synthesize the local control strategies so that it is ensured that the terrorists altogether manage to complete their mission of killing a few cartoonists. In the quantitative setting, we want the probability thereof to be at least a certain given value $p$.    \hnote{Hmhh.}\hnote{Hmhh.}
 
% \jnote{A reviewer thinks we did not try to persuade that our formalism is any good for modeling real problems. I think that we need a more aggressive example to sell the formalism.}

As an example, consider the problem of leaking a secret from a
sandboxed malware to an attacker. The behaviour of attacker and
malware (and possibly other components) are prescribed in terms of
states, private transitions, labelled synchronisation
transitions, and delay transitions  which model both local
computation times and synchronisation times. The delays are governed by
arbitrary continuous probability distributions over real time. % A
Handshake synchronisation is assumed to take place if all devices able
to do so agree on the same transition
label. %\jnote{is this previous sentence needed to explain the example? we have somthing similar later}
Otherwise the components run fully asynchronously.  The sandboxing can
be thought of as restricting the set of labels allowed to occur on
synchronisation transitions.
% \jnote{Would connect this paragraph with the next paragraph, this way, the example seems to be without any conclusion. I would shorten the next paragraph to tailor it towards to example. It does not make sense to me to list all we do before explaining the setting properly.}
 The  question we focus on is how to synthesise the
component control strategies for malware and attacker so that 
% it is ensured that 
they reach their target (of leaking the secret) almost
surely or with at least a given probability $p$.% These \emph{existential}
% questions determine whether there exist strategies with these
% properties. In addition, we also consider the \emph{value}
% problems. That is, to decide whether there are strategies that can
% reach their target with probability arbitrarily close to $p$ in the
% quantitative case, or to 1 in the qualitative case. 

More precisely, we consider a parallel composition of $n$
modules %$M_1,...,M_n$ 
synchronizing via handshake communication. %, representing the distributed system at hand.
The modules are modelled by interactive Markov chains (IMCs)~\cite{Her02,how-why}, a
generalization of labelled transition systems and of continuous time Markov
chains, equipped with a well-understood compositional theory.
Each module may in each state enable private actions, as well as
synchronisation actions. %This
%parallel composition will be called a \emph{distributed IMC}. 
It is natural to view such a \emph{distributed IMC} as a game
with $n+1$ players, where the last player controls the 
% nondeterminism due to 
interleaving of the modules. Each of the other $n$ players
controls the decisions in a single module, only based on 
its local timed history containing only transitions
that have occurred within the module.
%the local timed history of this module.  
On entering a state of its module, each player selects and commits to executing one of the actions enabled. A private action is executed immediately while
a synchronisation action requires a CSP-style handshake~\cite{CSP}, it
is executed once all modules able to perform this action have
committed to it. 
%All the while, each player only 
  
%This problem is a generalization of \hnote{...} and inspired by
%\hnote{....}.  \hnote{more here? Or simply drop this.}
  
%  \jnote{Would postpone the paragraph later; breaks the story a bit}
%\hnote{Disagree, since this appears in the example.}
For representing delay distributions, we make one decisive and one
technical restriction. First, we assume that each distribution is
continuous. 
This for instance disallows deterministic delays of, say, $3$ time units.
%This implies in particular that with probability 1 it
%cannot happen that a delay finishes at a precise time instant on the
%real time line. 
It is an important simplification assumed along our explorations of continuous-time distributed control. Second, we restrict to exponential distributions. This is a pure technicality, since (a) 
% the entirety of 
our results can be developed with general continuous
distributions, at the price of excessive notational and technical
overhead, and (b) exponential distributions can approximate arbitrary continuous distributions arbitrarily close~\cite{neuts1981matrix}. 
%Imposing the latter restriction 
Together, these assumptions enable us to work
in a setting close to interactive Markov chains.

%, and can do so based on the local timed history of this
%very module. So, each of these players only observes what happens in
%his own module and only interacts with other players by synchronisation
%actions. On entering a state of module $M_i$, player $i$ selects and
%commits to executing one of the action transitions enabled in this
%state. A private action is executed immediately while synchronisation
%actions require a handshake. Handshake communication uses the popular
% concept from CSP~\cite{CSP}: Transitions labelled
%with a synchronisation action $a$ are taken only if all other players
%that also have action $a$ in their set of actions do commit to this
%action at the same moment, the synchronisation then takes place
%immediately.
%instantly at this very moment. % just to save a orphan

Apart from running in continuous time, the concepts behind distributed
IMCs are rather common. Closely related are models based on
probabilistic automata~\cite{segala} or (partially observable) Markov
decision processes \cite{Paz71,BGIZ02}. In these settings, the power
of the interleaving player $n+1$ is a matter of ongoing
debate~\cite{switched-pioa,task-pioa,dargenio-pelozo}. The problem is
that without additional (and often complicated) assumptions this
player is too powerful to be realistic, and can for instance leak
information between the other players. This is a problem, e.g. 
% for a large number of applications 
in the security context, 
%and it can also 
making model checking results overly pessimistic~\cite{giro-dargenio}.

In sharp contrast to the discrete-time settings, in our distributed IMCs the interleaving player
$n+1$ does not have decisive influence on the resulting game.
%One of our main
%results shows that for distributed structures 
The reason  is
that 
% (under mild restrictions on winning conditions) 
the interleaving player
can only affect the order of transitions between two delays, but
neither \emph{which transitions} are taken nor what the different players \emph{observe}. 
%Thus, under mild restrictions on winning conditions, the interleaving player has no power.
This is rooted in  
the common alphabet synchronisation and
especially the continuous-time nature of the game: 
%inducing that 
the
probability of two local modules changing state at the same time
instant is zero, except if synchronising. %as well as an
%assumption of maximal progress when sequences of synchronisation are
%enabled.
%\jnote{edited the par above, okay?}
\begin{example}
\label{ex1}
    We consider the model displayed on the right where the
    delay transitions are labelled by some rate $\lambda$. It displays
    a very simplistic malicious $\app$ trying to communicate a
    secret to an outside $\att$acker, despite being
    sandboxed. Innocently
    looking action $\login$, $\logout$ and $\lookup$ synchronise $\app$ and $\att$, while the unlabelled transitions denote some private actions of the respective module.
\end{example}
%\vspace*{-1ex}

\begin{wrapfigure}[10]{r}{0.435\linewidth}
	%\begin{center}
	\vspace{-2.5em}
	\begin{tikzpicture}[yscale=0.45,xscale=0.45]
	\scriptsize
	
	\tikzstyle{every node}=[ellipse, draw=black,
	inner sep=0pt, minimum width=13pt, minimum height=13pt]
	
	\draw (0,8.3) node [draw=none] (app)	{\textbf{App:}};
	\draw (0,6) node (s0)	{$c_0$};
	\draw (6.5,6) node (s5)	{};
	\draw (2,7.2) node (l1)	{$t_1$};
	\draw (5,7.2) node (l2)	{$t_2$};
	\draw (8,7.2) node (l3)	{$t_3$};
	\draw (10,7.2) node (l4)	{$t_4$};
	\draw (2,4.8) node (r1)	{$b_1$};
	\draw (5,4.8) node (r2)	{$b_2$};
	\draw (8,4.8) node (r3)	{$b_3$};
	\draw (10,4.8) node (r4)	{$b_4$};
	
	\path[->] ($(s0)+(-1,0)$) edge (s0);
	
	\path[->] (s0) edge node [above left, draw=none] {$\lambda$} (l1);
	\path[->] (s0) edge node [below left, draw=none] {$\lambda$} (r1);
	
	\path[->] (l1) edge [loop above] node [above=-3, draw=none] {$\lambda$} (l1);
	\path[->] (l4) edge [loop above] node [above=-3, draw=none] {$\lambda$} (l4);
	\path[->] (l1) edge node [above=-3, draw=none] {\login} (l2);	
	\path[->] (l2) edge [] node [above=-2, draw=none] {$\lambda$} (s5);
	\path[->] (l2) edge node [above=-3, draw=none] {\lookup} (l3);
	\path[->] (l2) edge [bend right = 45] node [above=-3, draw=none] {\logout} (l1);
	\path[->] (l3) edge (l4);
	
	\path[->] (s5) edge [loop right] node [right=-3, draw=none] {$\lambda$} (s5);
	
	\path[->] (r1) edge [loop above] node [above=-3, draw=none] {$\lambda$} (r1);
	\path[->] (r4) edge [loop above] node [above=-3, draw=none] {$\lambda$} (r4);
	\path[->] (r1) edge node [below=-3, draw=none] {\login} (r2);	
	\path[->] (r2) edge [] node [below=-3, draw=none] {$\lambda$} (s5);
	\path[->] (r2) edge node [below=-3, draw=none] {\lookup} (r3);
	\path[->] (r2) edge [bend right = 45] node [above=-3, draw=none] {\logout} (r1);
	\path[->] (r3) edge (r4);
	
	\begin{scope}[xshift=-10cm,yshift=-4.4cm]
	\draw (10,7.5) node [draw=none] (app)	{\textbf{Att:}};
	
	\draw (12,6) node (t0)	{$\bar{c}_1$};
	\draw (15,6) node (t1)	{$\bar{c}_2$};
	\draw (18,6) node (t2)	{$\bar{c}_3$};
	\draw (20,7) node (t3)	{$\bar{t}_4$};
	\draw (20,5) node (t4)	{$\bar{b}_4$};
	
	\path[->] (t1) edge [bend left = 45] node [above=-3, draw=none] {$\lambda$} (t0);
	
	\path[->] ($(t0)+(-1,0)$) edge (t0);
	
	\path[->] (t0) edge [loop above] node [above=-3, draw=none] {$\lambda$} (t0);
	\path[->] (t0) edge node [above=-3, draw=none] {\login} (t1);
	\path[->] (t1) edge node [above=-3, draw=none] {\lookup} (t2);
%	\path[->] (t2) edge [bend right = 50] node [above=-3, draw=none] {\logout} (t0);
	\path[->] (t2) edge (t3);
	\path[->] (t2) edge (t4);
	
		\path[->] (t3) edge [loop above] node [left, draw=none] {$\lambda$} (t3);
	\path[->] (t4) edge [loop above] node [left, draw=none] {$\lambda$} (t4);
	\end{scope}
	\end{tikzpicture}
	%\end{center}
	%\caption{Transformation for states with synchronisation transitions that adds new private actions $b_1, \ldots, b_n$.}
	%	\label{fig:nondeterminism}
\end{wrapfigure}

%\mypara{Example.}
%     
%
%  	\jnote{The examples are long and are not structured to guide the reader. It would be nice to know which aspects are illustrated by which parts of the example. Somebody complained there is not much continuous-time in the examples at first sight -- I do agree. The modules also do not fit perfectly to the story.}
%\hnote{I agree and see no way to explain (1) the important features of the formalism and (2) the  importants aspects of the solution and (3) do this  on an intuitively convincing practically looking example in a (4) well-structured way and (5) take little space, all at the same time.}
%
%

	\vspace{-0.3em}
    Initially, the $\app$ can only let time pass. The $\att$acker
    player has no other choice than committing to handshaking on
    action $\login$. A race of the delay transitions will occur that at some point will lead to
    either state $(t_1,\bar{c}_1)$ or $(b_1,\bar{c}_1)$, with equal probability. Say in $(t_1,\bar{c}_1)$,
    the $\app$ player can only commit to action
    $\login$. The synchronisation will happen immediately since the
    $\att$acker is committed to $\login$ already, leading to $(t_2,\bar{c}_2)$. Now the $\app$ player has either to
    commit to action $\lookup$ or $\logout$. The latter will induce a
    deadlock 
%    after some time since $\logout$ is a synchronisation action. 
	due to a mismatch in players' commitments.
Instead assuming the earlier, the state synchronously and
    immediately changes to $(t_3,\bar{c}_3)$.
%    , because again the $\att$acker player has had no other choice than committing to it as well. 
	The $\att$acker player can now use its local timed
    history to decide which of the private actions to pick. Whatever
    it chooses, an interleaving of private actions of
    the two modules follows in zero time.
    Unless the reachability condition
%   cuts through zero timed sequences of events (such as the set $\{(t_3,t'_4)\}$ because $t_3$ is not at the end of a zero timed sequence), 
	considers transient states such as $(t_3,\bar{t}_4)$ where no time is spent,
the player
    resolving the interleaving 
%    (here and in all other cases) 
has no influence on 
the outcome.

    Now, assume the reachability condition is the state set
    $\{(t_4,\bar{t}_4),(b_4,\bar{b}_4)\}$. This
    corresponds to the $\att$acker player correctly determining the
    initial race of the $\app$, and can be considered as a leaked
    secret. However, according to the explanations provided, it should
    be obvious that the probability of guessing correctly (by
    committing properly in state $\bar{c}_3$) is no larger than 0.5, just
    because the players are bound to decide only based on the local
    history. The crucial question is: is there an algorithm to compute such probabilities, in general?
\vspace*{1ex}

  \mypara{Our contribution} This paper is the first to explore distributed
  cooperative reachability games with continuous-time flow modelled
  explicitly. The formalism we study is based on interactive Markov chains,
  which in turn has been applied across a wide range of engineering
  domains.
 We aim at synthesising local strategies for the players to
 reach with at least a given probability a specified set of goal states. If this probability is 1 we call the problem \emph{qualitative}, otherwise \emph{quantitative}. We consider \emph{existential} problems,
 asking for the existence of strategies with these properties, and
 \emph{value} problems, asking for strategies approximating the given
 probability value arbitrarily closely.
We have three main results:
  \begin{enumerate}
  	\item We show that, under mild assumptions on the winning condition, in continuous-time distributed synthesis the interleaving player has no power.
  	\item In general, we establish that the quantitative problems
          are undecidable for two or more players, the qualitative
          value problem is undecidable for two or more players and the
          qualitative existence problem is EXPTIME-hard for two
          players and undecidable for three or more players.
  	\item However, when focusing on the subclass of 2-player
          \emph{non-urgent} distributed IMCs, the
          quantitative value problem can be solved in
          exponential time. Non-urgency enables 
%          the possibility to
          changing the decisions committed to after some time. Thus, it
          empowers the players to reach a distributed consensus
          about the next handshake to perform by observing the only
          information they jointly have 
%          approximate 
		access to: the
          advance of time.
  \end{enumerate}
% 
%  The lower bounds and undecidability results originate from a reduction from decentralised partially observable Markov decision
%  processes (DEC-POMDP) to distributed IMCs that adds one player. DEC-POMDP are
%  multi-player extensions of 
%  POMDP.
%It is known that qualitative existence is decidable for POMDP~\cite{BGB12}, however,
%we show that qualitative existence is undecidable for DEC-POMDP already for 2 players. This result is interesting in itself. It is
%  to the knowledge of the authors the strongest undecidability result
%  for DEC-POMDPs with infinite horizon. This leads to undecidability
%  of qualitative existence for 3 or more players in distributed
%  IMCs.  
  
    The qualitative undecidability comes from a novel result about
    decentralised partially observable Markov decision processes (DEC-POMDP), a multi-player extensions of POMDP.
	While qualitative existence is decidable for POMDP~\cite{BGB12}, we show that qualitative existence is undecidable for DEC-POMDP already for 2 players. It is to the knowledge of the authors the strongest undecidability result for DEC-POMDPs with infinite horizon which is of its own interest.
	By a reduction from DEC-POMDP to distributed IMCs that adds one player, we get undecidability of qualitative existence for 3 or more players in distributed
    IMCs.  

%% file: defs.tex
\section{Distributed Interactive Markov Chains}\label{sec:defs}

%\jnote{In the intro: We should argue why everything is controlled in our setting -- our IMC components are not created by composition of smaller components, otherwise we would have problems to argue that we can control the interleaving non-determinism.} \hnote{I guess this will go to the discussion section as well, no?}
%---
%
%\jnote{changes: 
%	\begin{itemize}
%	\item
%	Removed initial state from the probability measure $\pr^{\vect{\sigma},\delta}$ as it is fixed in $\structure$. 
%	\item Removed notation for choices available in a state (have only words). 
%	\item Removed examples as they do not fit to the definitions as they are now. 
%	\item We probably need to illustrate a few concepts briefly on a few short examples (maybe a running example from the intro) -- I need that someone else first reads this.
%	\item Removed constraints checking consistency of a play -- it is now an arbitrary sequence with the right typing. Consistency is checked on the level of the probability measure.
%	\item I have surely forgotten to define a few minor notions, names, etc. Feel free to return back any small things that are needed.
%\end{itemize}	}

We denote by $\Rset$, $\Rsetpo$, $\Nset$, and $\Nseto$ the sets of real numbers, non-negative real numbers, positive integers, and non-negative integers, respectively. Furthermore, for a finite set $X$, we denote by $\dist(X)$ the set of \emph{discrete probability distributions} over $X$, i.e. functions $f:X\to[0,1]$ such that $\sum_{x\in X} f(x) = 1$.
Finally, for a tuple $\vect{x}$ from a product space $X_1 \times \cdots \times X_n$ and for $1 \leq i \leq n$, we use functional notation $\vect{x}(i)$ to denote the $i$th element of the tuple.

We first give a definition of a (local) module based on the formalism of Interactive Markov Chains (IMC). Then we introduce (global) distributed IMC. %\footnote{It is closely related to the formalism of interactive Markov chains.} 
%
%\jnote{After presenting it several times, I do not like ``Distributed structures'' very much. It is too vague. What about coming back to IMC and calling it ``distributed IMC'' and ``IMC (component)''?}
\begin{mydef}[IMC]
An \emph{IMC} (module) is a tuple $(S, \act, \inter{}, \marko{}, \initstate)$ where

\begin{itemize}

\item $S$ is a finite set of states with an initial state $\initstate$,
\item $\act$ is a finite set of actions,
\item $\inter{} \subseteq S \times \act \times S$ is the \emph{action} transition relation,
\item $\marko{} \subseteq S \times \mathbb{Q}_{> 0} \times S$ is the finite \emph{delay} transition relation.
\end{itemize}
\end{mydef}

\noindent
We write $s \intersmall{a} s'$ when $(s,a,s') \in \inter{}$ and $s \markosmall{\lambda} s'$ when $(s,\lambda, s') \in \marko{}$ ($\lambda$ being the \emph{rate} of the transition). We say that action $a$ is \emph{available} in $s$ if $s \intersmall{a} s'$ for some~$s'$.
%This is to avoid infinite sequences of duration 0. 
%We write $\succes_{\imc}(s)$ for the set of states $s'$ in IMC $\imc$ such that exists an action $a$ with $s \inter{a} s'$ or if there exists a rate $\lambda$ with $s \marko{\lambda} s'$. We are now ready to define a distributed structure.

\begin{mydef}[Distributed IMC]
A \emph{distributed IMC}  is a tuple $$\structure = ((S_i, \act_i, {\inter{}}_i, {\marko{}}_i, \initstate_{i}))_{1 \leq i \leq n}$$ of modules for players $\agt = \{1,...,n\}$. 
Furthermore, by $\act = \bigcup_{i} \act_i$ we denote the set of all actions, and by $S = S_1 \times ... \times S_n$ the set of (global) states.
\end{mydef}

Intuitively, a distributed IMC moves in continuous-time from a (global) state to a (global) state using transitions with labels from $\lab = \act \cup \agt$: 
%
%where action transitions have label from $\act$ and delay transitions have label from $\agt$ identifying the player taking the transition:
\begin{itemize}
	\item An action transition with label $a \in \act$ corresponds to \emph{synchronous} communication of all players in $\syn(a) := \{j \in \agt \mid a \in \act_j \}$ and can only be taken when it is \emph{enabled}, i.e. when all these players choose their local transitions with action $a$ at the same time.
It is called a \emph{synchronisation} action if $|\syn(a)| \geq 2$ and a \emph{private} action if $|\syn(a)| = 1$.
	\item A delay transition of any player $j \in \agt$ is taken \emph{independently} by the player after a random delay, i.e. the set of players that synchronise over label $j$ is $\syn(j) = \{j\}$.
\end{itemize}
%Action transition is assumed to take zero time and thus, delay transition can be taken only if no action is enabled.

%\jnote{this is a tough paragraph here}
Formally, the (local) \emph{choices of player $j$} range over $\choice_j = \inter{}_j \cup \{\bot\}$. When in (local) state $s$, the player may pick only a choice \emph{available in $s$}. That is, either an action transition of the form $s \intersmall{a} s'$ or $\bot$ if there is no such action transition.
We define \emph{global choices} as $\choice = \choice_1 \times \cdots \times \choice_n$. 
A global choice $\vect{c}$ induces the set $\enab(\vect{c}) = \{a \in \act \mid \forall j \in \syn(a) : \vect{c}(j) = (\cdot, a,\cdot) \}$ of actions  \emph{enabled} in $\vect{c}$.

To avoid that time stops by taking infinitely many action steps in zero time, 
we pose a standard assumption
 prohibiting cycles~\cite{HKK13,conf/mmb/HermannsJ08,imca,mrmc,Kat09b}: we require that for every action $a\in\act$ there is a player $j\in\syn(a)$ such that the labelled transition system $(S_j,\inter{}_j)$ does not have any cycle involving action $a$.

%
%Furthermore, to have a unified treatment of action and delay transitions, we use $\lab = \act \cup \agt$ as a domain of transition labels. Here, $j\in\agt$ identifies a delay transition of player $j$. Finally, we extend $\syn(a)$ to any $a\in \lab$ by $\syn(j) := \{j\}$ for any $j \in \agt$.

The behaviour of a distributed IMC is a \emph{play}, an infinite sequence
$$\rho = \vect{s_0} \vect{c_0} \arrow{a_1,t_1}{} \vect{s_1} \vect{c_1} \arrow{a_2,t_2}{} \vect{s_2} \vect{c_2}  \cdots $$
where each $\vect{s_i} \in S$ is the state after $i$ moves, $\vect{c_i} \in \choice$ is the choice of the players in the state $\vect{s_i}$, and $a_{i+1} \in \lab$ and $t_{i+1} \in \Rsetpo$ are the label and the absolute time of the next transition taken. By $\play$ we denote the set of all plays. 
Which play is taken depends on the \emph{strategies} of the players, on the \emph{scheduler} which resolves interleaving of communication whenever multiple actions are enabled, and on the rules (involving randomness) given later.

\subsection{Schedulers and strategies}
First we define strategies and schedulers basing their decision on the current local and global history, respectively. A \emph{(global) history} is a finite prefix 
$$h = \vect{s_0} \vect{c_0} \arrow{a_1,t_1}{} \cdots \arrow{a_{i},t_{i}}{} \vect{s_i}$$
of a play ending with a state. For given $h$, we get the \emph{local history} of player $j$ as
$$\pi_j(h) = \vect{s'_{0}}(j) \vect{c'_{0}}(j) \arrow{a'_1,t'_1}{} \cdots \arrow{a'_{\ell}, t'_{\ell}}{} \vect{s'_{\ell}}(j)$$
where $\vect{s'_0} \vect{c'_0} \arrow{a'_1,t'_1}{} \cdots \arrow{a'_{\ell}, t'_{\ell}}{} \vect{s'_\ell}$ is the subsequence of $h$ omitting all steps not visible for player $j$, i.e. all $\arrow{a_{m}, t_{m}}{} \vect{s_m} \vect{c_m}$ with $j \not\in\syn(a_m)$. The set of all global histories is denoted by $\hist$; the set of local histories of player $j$ by $\hist_j$. 

\begin{example}\label{ex2}
 Consider again Example~\ref{ex1}. Let $\app$ be controlled by player 1 and $\att$ by player 2. For the following history we get corresponding local histories 
 \begin{align*}
h &= (c_0,\bar{c}_1) (\bot, \login) \arrow{1, 0.42}{} (t_1,\bar{c}_1) (\login, \login) \arrow{\login, 0.42}{} (t_2,\bar{c}_2), \\
\pi_1(h) &= c_0~ \bot \arrow{1, 0.42}{} t_1 ~ \login \arrow{\login,0.42}{} t_2,
\qquad \pi_2(h) = \bar{c}_1 ~\login \arrow{\login,0.42}{} \bar{c}_2
 \end{align*}
Note that the attacker can neither observe the Markovian transition nor the local state of the \app{}. The \app{} cannot observe the local state of the attacker either, but it can be deduced from the local history of the \app{}.
\end{example}
%\mypara{Example.}

A \emph{strategy} for player $j$ is a measurable
%\footnote{In~\citeapp{app:defs}, we give the measurable spaces to clarify the use of ``measurable''.}
function $\sigma: \hist_j \rightarrow \dist(\choice_j)$ that assigns to any local history $h$ of player $j$ a probability distribution over choices available in the last state of $h$. We say that a strategy $\sigma$ for player $j$ is 
%\begin{itemize}
%	\item 
\emph{pure} if for all $h$ we have $\sigma(h)(c) = 1$ for some $c$; and
	%	\item \emph{untimed} if for any $h$ and $h'$ that differ only in the delays we have $\sigma(h) = \sigma(h')$; and
%\item 
\emph{memoryless} if for all $h$ and $h'$ with equal last local state we have $\sigma(h) = \sigma(h')$.
%\end{itemize}
%
%We can view a pure strategy as a function $\sigma: \hist_j \to \choice_j$ and write $\sigma(h) = c$. Similarly, we can view a memoryless strategy as a function $\sigma: S_j \to \dist(\choice_j)$. \jnote{do I need this?}

A \emph{scheduler} is a measurable function  $\delta: \hist \times \choice \to \dist(\act) \cup \{\bot\}$ that assigns to any global history $h$ and global choice $\vect{c}$ a probability distribution over actions enabled in $\vect{c}$; or a special symbol $\bot$ again denoting that no action is enabled.

\begin{example}
	The available local choices in $(t_2,\bar{c}_2)$, the last state of $h$ from above, are $\{(t_2, \mathit{lookup}, t_3), (t_2, \mathit{logout}, t_1) \}$ for $\app$  and solely $\{(\bar{c}_2, \mathit{lookup}, \bar{c}_3)\}$ for $\att$. Let the strategy of $\app$ select either choice with equal probability. 
If $(t_2, \mathit{lookup}, t_3)$ is chosen, $\mathit{lookup}$ is enabled and must be picked by the scheduler $\sigma$. 
If $(t_2, \mathit{logout}, t_1)$ is chosen, no action is enabled and $\delta$ must pick $\bot$, waiting for a delay transition.
\end{example}

\subsection{Probability of plays} Let us fix a \emph{profile of strategies} $\vect{\sigma} = (\sigma_1,\ldots,\sigma_n)$ for individual players, and a scheduler $\delta$. The play starts in the initial state $\vect{s_0} = (\initstate_1, \ldots ,\initstate_n)$  and inductively evolves as follows. Let the current history be $h = \vect{s_0} \vect{c_0} \arrow{a_1,t_1}{} \cdots \arrow{a_{i},t_{i}}{} \vect{s_i}$.
\begin{itemize}
	\item For the next choice $\vect{c_i}$, only players $P_i := \syn(a_i)$ involved in the last transition freely choose (we assume $P_0 := \agt$). Hence, independently for every $j \in P_i$, the choice $\vect{c_i}(j)$ is taken randomly according to $\sigma_j(\pi_j(h))$. All remaining players $j \not\in P_i$ stick to the previous choice $\vect{c_i}(j) = \vect{c_{i-1}}(j)$ as for them, no observable event happened. 
	\item After fixing $\vect{c_i}$, there are two types of transitions:
	\begin{enumerate}
		\item If $\enab(\vect{c_i}) \neq \emptyset$, the next synchronization action $a_{i+1} \in \enab(\vect{c_i})$ is chosen randomly according to $\delta(h,\vect{c_i})$ and taken immediately at time $t_{i+1} := t_i$.
		The next state $\vect{s_{i+1}}$ satisfies for every $j \in \agt$:
		$$
		\vect{s_{i+1}}(j) = \begin{cases}
		\mathrm{target}(\vect{c_i}(j))  & \text{if $j \in \syn(a_{i+1})$,} \\
		\vect{s_{i}}(j) & \text{if $j \not\in \syn(a_{i+1})$.}
		\end{cases}
		$$
		where $\mathrm{target}(\vect{c_i}(j))$ denotes the target of the transition chosen by player $j$.
		In other words, players involved in synchronisation move according to their choice, the remaining players stay in their previous states.
		\item If $\enab(\vect{c_i}) = \emptyset$, a local delay transition is taken after a random delay, chosen as follows. Each delay transition $\vect{s_i}(j) \markosmall{\lambda} \cdot$ outgoing from the current local state of any player $j$ is assigned randomly a real-valued delay according to the exponential distribution with rate $\lambda$. This results in a collection of real numbers. The transition $\vect{s_i}(\ell) \markosmall{\lambda} s$ with the minimum delay $d$ in this collection is taken. 
%		
%		$\ell$ be the player and $\vect{s_i}(\ell) \marko{\lambda} s$ be the transition such that the delay $d$ of this transition is minimal in this collection. 
		Hence, $a_{i+1} := \ell$ (denoting that player $\ell$ moves), $t_{i+1} := t_i + d$, and the next state $\vect{s_{i+1}}$ satisfies for every $j \in \agt$:
		$$
		\vect{s_{i+1}}(j) = \begin{cases}
		s  & \text{if $j \in \syn(a_{i+1}) = \{\ell\}$,} \\
		\vect{s_{i}}(j) & \text{if $j \not\in \syn(a_{i+1})$.}
		\end{cases}
		$$
	\end{enumerate}
\end{itemize}

\noindent
All these rules induce a probability measure $\pr^{\vect{\sigma},\delta}$ over the set of all plays by a standard cylinder construction. %For details, see~\citeapp{app:defs}.

\subsection{Distributed synthesis problem}
We study the following two fundamental reachability problems for distributed IMCs.
	Let $\structure$ be a distributed IMC, $T \subseteq S$ be a target set of states, and $p$ be a rational number in $[0,1]$.
Denoting by $\diamond T$ the set of plays $\rho$ that reach a state in $T$ and stay there for a non-zero amount of time, we focus on:
	\begin{description}
		\item[Existence] Does there exist a strategy profile $\vect{\sigma}$ s.t. for all schedulers $\delta$,		
		$$ \; \pr^{\vect{\sigma},\delta}(\diamond T) \geq p \; ?$$
		\item[Value] Can the value $p$ be arbitrarily approached, i.e. do we have
		$$\;\sup_{\vect{\sigma}} \inf_\delta \pr^{\vect{\sigma},\delta}(\diamond T) \geq p \; ?$$
	\end{description}	
	We refer to the general problem with $p \in [0,1]$ as \emph{quantitative}. When we restrict to $p=1$, we call the problem \emph{qualitative}.

%
%\subparagraph{Distributed synthesis problem}
%%
%We study the following fundamental reachability problems for distributed structures.
%%
%For a set $T \subseteq S$, let $\diamond T$ denote the set of plays $\rho$ that reach a state in $T$ and stay there for a non-zero amount of time.
%
%\begin{myproblem}
%	Let $\structure$ be a distributed structure, $T \subseteq S$ be a target set of states, and $p$ be a rational number in $[0,1]$. We study two questions:
%	\begin{description}
%		\item[Existence] Does there exist a strategy profile $\vect{\sigma}$ s.t. for all schedulers $\delta$,
%		$ \; \pr_{\vect{s_0}}^{\vect{\sigma},\delta}(\diamond T) \geq p \;$?
%		\item[Value] Can the value $p$ be arbitrarily approached, i.e. do we have $\;\sup_{\vect{\sigma}} \inf_\delta \pr_{\vect{s_0}}^{\vect{\sigma},\delta}(\diamond T) \geq p \; $?
%	\end{description}	
%%
%We refer to the problems as \emph{qualitative} if $p = 1$  and as \emph{quantitative} if $p < 1$.
%
%\jnote{check if std. literature defines quantitative by $p <1$ or by $p \leq 1$.}
%\end{myproblem}

%% file: results-schedulers.tex
\section{Schedulers are not that powerful}\label{sec:results-schedulers}

The task of a scheduler is to choose among concurrently enabled
transitions, thereby resolving the non-determinism conceptually caused
by interleaving. In this section, we address 
%the question ``what 
the impact of the decisions of the scheduler.
% can have?''. Only little, actually.
%
We 
%will 
show that despite having the ability to affect the order in which transitions are taken in the \emph{global} play, the scheduler cannot affect what every player observes \emph{locally}.
Thus, the scheduler affects neither the choices of any player nor what synchronisation occurs. 
%This means that 
%If we consider 
As a result,
for winning objectives that are closed under 
%permutations preserving local observations
\emph{local observation equivalence}, the scheduler cannot affect the probability of winning.
\begin{wrapfigure}[6]{r}{0.35\linewidth}
	%\begin{center}
	\vspace{-1.7em}
	\begin{tikzpicture}[yscale=0.35,xscale=0.40]
	\scriptsize
	
	\tikzstyle{every node}=[ellipse, draw=black,
	inner sep=0pt, minimum width=11pt, minimum height=11pt]
	
	\draw (2,9) node [label=below right:$s_0$] (s0)	{$\bullet$};
	\draw (2,7) node [label=below right:$s_1$] (s1)	{};
	\draw (2,5) node [label=below right:$s_2$] (s2)	{};
	\draw (2,3) node [label=below right:$s_3$] (s3)	{};

	\draw (5,9) node [label=below right:$t_0$] (t0)	{$\bullet$};
	\draw (5,7) node [label=below right:$t_1$] (t1)	{};
	\draw (5,5) node [label=below right:$t_2$] (t2)	{};

	\draw (8,9) node [label=below right:$u_0$] (u0)	{$\bullet$};
	\draw (8,7) node [label=below right:$u_1$] (u1)	{};

	\draw (11,9) node [label=below right:$v_0$] (v0)	{$\bullet$};
	\draw (11,7) node [label=below right:$v_1$] (v1)	{};

	\draw (1.5,10) node [draw=none] (c1)	{$C_1$};
	\draw (4.5,10) node [draw=none] (c2)	{$C_2$};
	\draw (7.5,10) node [draw=none] (c2)	{$C_3$};
	\draw (10.5,10) node [draw=none] (c2)	{$C_4$};

	\path[->] (s0) edge node [left, draw=none] {$\lambda$} (s1);
	\path[->] (s1) edge node [left, draw=none] {$a$} (s2);
	\path[->] (s2) edge node [left, draw=none] {$b$} (s3);

	\path[->] (t0) edge node [left, draw=none] {$a$} (t1);
	\path[->] (t1) edge node [left, draw=none] {$c$} (t2);

	\path[->] (u0) edge node [left, draw=none] {$b$} (u1);

	\path[->] (v0) edge node [left, draw=none] {$c$} (v1);
	\end{tikzpicture}
	%\end{center}
	%\caption{Transformation for states with synchronisation transitions that adds new private actions $b_1, \ldots, b_n$.}
	%	\label{fig:nondeterminism}

	\end{wrapfigure}

\begin{example}
 Consider the distributed IMC to the right. After the delay transition is taken in $C_1$ and there is synchronisation on action $a$, the scheduler can choose whether there will be synchronisation on $b$ or $c$ first. However, it can only affect the interleaving, not any of the local plays.
\end{example}

For a play $\rho = \vect{s_0} \vect{c_0} \arrow{a_1,t_1}{} \vect{s_1} \vect{c_1} \arrow{a_2,t_2}{} \vect{s_2} \vect{c_2}  \cdots $ we define the local play $\pi_j(\rho)$ of player $j$ analogously to local histories.\mycut{ It is obtained by removing all steps $\arrow{a_{m}, t_{m}}{} \vect{s_m} \vect{c_m}$ where the player is not involved, i.e. $j \not\in\syn(a_m)$, and by replacing in the subsequence all global states and choices by the local states and choices of player $j$.} We define \emph{local observation} equivalence $\sim$ over plays by setting $\rho \sim \rho'$ if $\pi_j(\rho) = \pi_j(\rho')$ for all $j\in \agt$. Let us stress that two local observation equivalent plays have exactly the same action and delay transitions happening at the same moments of time; only the order of action transitions happening at the same time can differ.
Finally, we say that a set $E$ of plays is \emph{closed under local observation equivalence} if for any $\rho \in E$ and any $\rho'$ such that $\rho \sim \rho'$ we have $\rho' \in E$. It is now possible to show the following.

\begin{mytheo}
	\label{theo:sched_power}
	Let $E$ be a measurable set of plays closed under local observation equivalence. For any strategy profile $\vect{\sigma}$ and schedulers $\delta$ and $\delta'$ we have
	$$\pr^{\vect{\sigma}, \delta}(E) = \pr^{\vect{\sigma}, \delta'}(E).$$
\end{mytheo}
\mycut{
\begin{proof}[Sketch]
The proof of Theorem \ref{theo:sched_power} is a bit technical and therefore relegated to the appendix. It can be divided into three main parts. The overall idea of the three parts are
\begin{enumerate}
	\item Showing that for pure strategy profiles, the schedulers cannot affect anything but the interleaving in a maximal 0-duration sequence of actions;
	
	\item A proof showing that 1. can be lifted to all events $E$ closed under observation equivalence. This is done in two steps. First, it is shown to be the case for all interleaving abstract cylinders by induction using 1. Second, standard results from measure theory imply the result for all measurable sets of plays closed under local observation equivalence;
	
	\item Extending 2. from pure strategies to arbitrary strategies can be done reusing ideas from classical game theory on mixed and behavioural strategies in extensive-form games \cite{Kuhn53}.
\end{enumerate}
\end{proof}
}

As a result, for the rest of the paper we write $\pr^{\vect{\sigma}}(E)$ instead of $\pr^{\vect{\sigma},\delta}(E)$ since the scheduler cannot affect the probability of the events we consider. Indeed, the reachability objectives defined in the previous section are closed under local observation equivalence.\mycut{ Indeed, if a non-empty interval $[t,t')$ of time is spent in a target state $(s_1,\ldots,s_n) \in T$, the local play of every player $j$ has to visit $s_j$ at time $\leq t$ and leave it at time $\geq t'$. Then, every global play with such local plays belongs to $\diamond T$.}

\begin{remark}
%The fact that an interleaving scheduler does not have much impact in continuous-time may seem to be unsurprising. Yet, in this form, it does not hold for many small variations of the definitions of the formalism that we explored, e.g. for asymmetric communication or when allowing cycles of 
%action transitions.
The fact that interleaving does not have decisive impact in continuous
time may seem natural and thus possibly unsurprising to
experts. Yet, the result does not hold for many small variations
of the setting we consider, e.g. neither for asymmetric communication
nor when allowing cycles of action transitions.
\end{remark}

%% file: results-reductions.tex
\section{Undecidability Results}\label{sec:results-reductions}

In this section, we put distributed IMCs into context of other partial-observation models. 
As a result, we show that reachability quickly gets undecidable here.
% for distributed IMCs.

\begin{mytheo}\label{thm:undecidable}
	For distributed IMCs we have that
	\begin{enumerate}
		\item the qualitative value, quantitative value, and quantitative existence problems are undecidable with $n \geq 2$ players; and
		\item the qualitative existence problem is $\exptime$-hard with $n = 2$ players and undecidable with $n \geq 3$ players.
	\end{enumerate}
\end{mytheo}

Theorem~\ref{thm:undecidable} is obtained by using two fundamental results. First, we provide a novel (and somewhat surprising) result for \emph{decentralized POMDPs} (DEC-POMDPs) \cite{BGIZ02}, an established multi-player generalization of POMDPs. We show that the qualitative existence problem for DEC-POMDPs is undecidable already for $2$ players. This is, to the knowledge of the authors, currently the strongest known undecidability result for DEC-POMDPs. Second, we show that distributed IMCs are not only more expressive (w.r.t. reachability) than POMDPs but also more expressive than DEC-POMDPs. We show it by reducing reachability in
%the \emph{discrete-time}
DEC-POMDPs with $n$ players to reachability in
%the \emph{continuous-time}
distributed IMCs with $n+1$ players. Theorem~\ref{thm:undecidable} follows from these two results and from known results about POMDPs \cite{Paz71,BMT77,GO10}. For an overview, see Table \ref{table:undec}.%This gives us undecidability for distributed IMCs of $\geq 3$ players. As POMDPs are 1-player DEC-POMDPs the theorem now follows from undecidability of the quantitative existence problem~\cite{Paz71}, the quantitative value problem~\cite{BMT77}, as well as the qualitative value problem~\cite{GO10} in POMDPs as well as $\exptime$-hardness of the qualitative existence problem for POMDPs.

\begin{table}[!t]
\renewcommand{\arraystretch}{1.2}

\centering
\begin{tabular}{@{\;}l@{\;\;}|@{\;\;}c@{\;\;}|@{\;\;}c@{\;\;}|@{\;\;}c@{\;}}
\hline
 & POMDPs & DEC-POMDPs & Distributed IMCs\\
\hline\hline
Qual. Existence & Dec. \cite{BGB12} & Undec. for $\geq 2$ players & Undec. for $\geq 3$ players \\
\phantom{Qual.} Value & Undec. \cite{GO10} & Undec. for $\geq 1$ player \cite{GO10} & Undec. for $\geq 2$ players \\
\hline
Quant. Existence & Undec. \cite{Paz71} & Undec. for $\geq 1$ player \cite{Paz71} & Undec. for $\geq 2$ players \\
\phantom{Quant.} Value & Undec. \cite{BMT77} & Undec. for $\geq 1$ player \cite{BMT77} & Undec. for $\geq 2$ players
\end{tabular}
\vspace*{2mm}
\caption{Undecidability results for reachability. Unreferenced results are shown here.}
\label{table:undec}
\end{table}

\subsection{Decentralized POMDP (DEC-POMDP)}

We start with a definition of the related formalism of decentralized POMDP \cite{BGIZ02}.

\begin{mydef} A DEC-POMDP is a tuple $(S, \agt, (\act_i, \obs_i)_{1 \leq i \leq n}, \transfunc, \obsfunc, \initstate)$ where

\begin{itemize}
 \item $S$ is a finite set of global states with initial state $\initstate \in S$,
 
 \item $\agt = \{1,...,n\}$ is a finite set of players,
 
 \item $\act_i$ is a finite set of local actions of player $i$ with $\act_i \cap \act_j = \emptyset$ if $j \neq i$, (by $\act = \act_1 \times \cdots \times \act_n$ we denote the set of global actions),
 
  \item $\obs_i$ is a finite set of local \emph{observations} for player $i$, \\
  (by $\obs = \obs_1 \times \cdots \times \obs_n$ we denote the set of global observations),
 
 \item $\transfunc: S \times \act \rightarrow \Delta(S)$ is the transition function which assigns to a state and a global action a probability distribution over successor states, and
 
 \item $\obsfunc: S \times \act \times S \rightarrow \Delta(\obs)$ is the \emph{observation function} which assigns to every transition a probability distribution over global observations.

\end{itemize}
 
\end{mydef}

%A \emph{play} of a DEC-POMDP is an infinite sequence $\rho = s_0 \vect{a_0} s_1 \vect{a_1} \cdots$ where $s_0 = \initstate$ and for all $i \geq 0$ we have $s_i \in S$ and $\vect{a_i}\in \act$. Then an infinite number of rounds is played, in each round every player $j$ chooses an action $a_j$ from $\act_j$. Based on the current state $s$ and the tuple $\vect{a} = (a_1,\dots,a_n)$ of actions, a successor state $s'$ is 
%chosen randomly with probability $\delta(s,\vect{a})(s')$ and then a tuple of observations $\vect{o} = (o_1,\dots,o_n)$ of the players is chosen randomly with probability $\obsfunc(s,\vect{a},s')(\vect{o})$.

In contrast to distributed IMCs that capture flow of time explicitly, DEC-POMDP is a discrete-time formalism.
A DEC-POMDP starts in the initial state $\initstate$. Assuming that the current state is $s$, one discrete step of the process works as follows.
%
%It proceeds for an infinite number of discrete-time steps as follows, assuming the current state is $s$: 
%
First, each player $j$ chooses an action $a_j$. Then the next state $s'$ is chosen according to the probability distribution $\transfunc(s, \vect{a})$ where $\vect{a} = (a_1,\ldots,a_n)$. Then, each player $j$ receives an observation $o_j \in \obs_j$ such that the observations $\vect{o} = (o_1,...,o_n)$ are chosen with probability $\obsfunc(s, \vect{a},s')(\vect{o})$. Repeating this forever, we obtain a \emph{play} which is an infinite sequence $\rho = s_0 \vect{a_0} \vect{o_0} s_1 \vect{a_1} \vect{o_1} \cdots$ where $s_0 = \initstate$ and for all $i \geq 0$ it holds that $s_i \in S$, $\vect{a_i}\in \act$, and $\vect{o_i} \in \obs$.
%determined by a random choice according to the transition function $\delta$ such that state $s'$ is chosen with probability $\delta(s,\vect{a})(s')$. Next, based on $s, s'$ and $\vect{a}$, an observation $o_j \in \obs_j$ for each player $j$ is chosen randomly according to the observation function $\obsfunc$. It assigns probability $\obsfunc(s,\vect{a},s')(\vect{o})$ to the observation tuple $\vect{o} = (o_1,...,o_n)$.
%
%It is important that 
%
Note that the players can only base their decisions on the sequences of observations they receive rather than the actual sequence of states which is not available to them. 
%Though, the observations can give the players knowledge about the current state. 
%
For a more complete coverage of DEC-POMDPs, see~\cite{BGIZ02}.

\subsection{Reduction from DEC-POMDP}

First we present the reduction from a DEC-POMDP $\pomdp$ to a distributed IMC $\img$. In this subsection, we write $\pr^{\vect{\sigma}}_\pomdp$ or $\pr^{\vect{\sigma}}_\img$ instead of $\pr^{\vect{\sigma}}$ to distinguish between the probability measure in the DEC-POMDP from the probability measure in the distributed IMC. 

\begin{myprop}
\label{prop:reduc_decpomdp}
 For a DEC-POMDP $\pomdp$ with $n$ players and a target set $T$ of states of $\pomdp$ we can construct in polynomial time a distributed IMC $\img$ with $n+1$ players and a target set $T'$ of global states in $\img$ where:
 $$\exists \vect{\sigma}:  \pr^{\vect{\sigma}}_\img(\diamond T) = p 
 \;\; \iff \;\; 
 \exists \vect{\sigma}': \pr^{\vect{\sigma}'}_\pomdp(\diamond T') = p.
 $$
% 
% there exists a strategy profile $\sigma$ in $\pomdp$ with $\pr^\sigma_{s_0}(\diamond T) = p$ if and only if there exists a strategy profile $\sigma'$ in $\img$ with $\pr^{\sigma'}_{s'_0}(\diamond T') = p$

 %Deciding whether there is a strategy profile in a DEC-POMDP $\pomdp$ with $n$ players that satisfies a linear-time property $\obj \subseteq S^\omega$ with probability at least $p$ can be reduced to deciding whether there is a strategy profile satisfying a similar property $\obj'$ with probability at least $p$ \snote{Make more precise what similar means}\hnote{yes!} in a distributed IMC $\img$ of size polynomial in the size of $\pomdp$ and with $n+1$ players.
 
\end{myprop}

\begin{proof}[Proof sketch]
Let us fix $n$ and $\pomdp = (S, \agt, (\act_i)_{1 \leq i \leq n}, \delta, (\obs_i)_{1 \leq i \leq n}, \obsfunc)$ where $\agt = \{1,...,n\}$. 
Further, let $\act_i = \{a_{i1},...,a_{im_i} \}$ and $\obs_i = \{o_{i1},...,o_{i\ell_i} \}$ for player $i \in \agt$. 
The distributed IMC $\img$ has $n+1$ modules, one module for each player in $\pomdp$ and the \emph{main} module responsible for their synchronisation. Intuitively,
\begin{itemize}
	\item the module of every player $i$ stores the last local observation in its state space. Every step of $\pomdp$ is modelled as follows: The player \emph{outputs} to the main module the action it chooses and then \emph{inputs} from the main module the next observation.
	\item The main module stores the global state in its state space. Every step of $\pomdp$ corresponds to the following: The main module \emph{inputs} the actions of all players one by one, then it randomly picks the new state and new observations according to the rules of $\pomdp$ based on the actions collected. The observations are lastly \emph{output} to all players, again one by one.
\end{itemize}

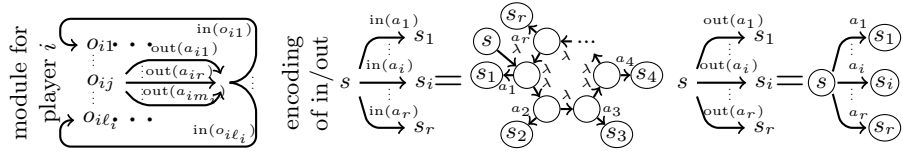
\begin{figure}
	\begin{center}
		\begin{tikzpicture}[xscale=0.7,yscale=0.5]

		\tikzstyle{state}=[ellipse, draw=black,
		inner sep=0pt, outer sep=1pt, minimum width=10pt, minimum height=10pt]
		\tikzstyle{coord}=[inner sep = 0, outer sep=0]
		\tikzstyle{dots}=[yscale=0.5, xscale=0.5]
		\tikzstyle{bigdots}=[yscale=1.6, xscale=1.6]
		\tikzstyle{action}=[font=\tiny,->,thick, rounded corners]
		\tikzstyle{rate}=[font=\tiny,text width=2cm,text centered]
		\tikzstyle{brace}=[decoration={brace, mirror},
		decorate,font=\footnotesize]

		\begin{scope}[xshift=0cm,yshift=-6cm,xscale=0.6]
		
		\node[rotate=90,font=\footnotesize,text width=2cm,text centered] at (-2.6,0) {module for player $i$};
		
		\begin{scope}[xshift=-0.5cm]
		\node[] (s1) at (0,1) {$o_{i1}$};
		\node[bigdots] at (1.1,1) {$\cdots$};
		\node[dots] at (-0.25,0.7) {$\vdots$};
		\node[] (si) at (0,0) {$o_{ij}$};
		\node[dots] at (-0.25,-.3) {$\vdots$};
		\node[] (sn) at (0,-1) {$o_{i \ell_i}$};
		\node[bigdots] at (1.1,-1) {$\cdots$};
		\end{scope}

		\begin{scope}[xshift=0cm]
		\node[] (si) at (0,0) {};
		\node[] (i1) at (3.4,0) {};
		
		\path[action] (si) edge node[above=-2,pos=0.6] {out($a_{i r}$)} (i1);
		
		\draw[action] (si.north east) -- (0.7,0.6) -- node[above=-2,pos=0.6] {
			out($a_{i 1}$)
		} (3.1,0.6) -- (i1) ;
		
		\draw[action] (si.south east) -- (0.7,-0.6) -- node[above=-3,pos=0.6] {
			out($a_{i m_i}$)
		} (3.1,-0.6) -- (i1) ;
		
		\node[dots] at (0.8,-.2) {$\vdots$};
		\node[dots] at (0.8,.4) {$\vdots$};
		\end{scope}
		
		\begin{scope}[xshift=3.4cm]
		\node[] (si) at (0,0) {};
		
		\draw[action] (si.east) -- (0.5,0) -- (1,0.5) |- (0,1.7) -| node[below=-2,pos=0] {
			in($o_{i 1}$)
		} ($(s1)+(-1.3,0)$) -- (s1) ;
		
		\draw[action] (si.east) -- (0.5,0) -- (1,-0.5) |- (0,-1.7) -| node[above=-2,pos=0] {
			in($o_{i \ell_i}$)
		} ($(sn)+(-1.3,0)$) -- (sn) ;
		
		\node[dots] at (0.9,0.1) {$\vdots$};
		\end{scope}
		
		\end{scope}

		\begin{scope}[xshift=4.5cm,yshift=-6cm,xscale=0.8]
		
		\node[rotate=90,font=\footnotesize,text width=2cm,text centered] at (-1.1,0) {encoding of in/out};
		
		\begin{scope}[xscale=0.7,xshift=-0.3cm]
		
		\node (s) at (0,0) {$s$};
		
		\node[] (si) at (2.7,0) {$s_i$};
		\node[] (sr) at (2.7,-1.2) {$s_r$};
		\node[] (s1) at (2.7,1.2) {$s_1$};
		
		\path[action] (s) edge node[above,pos=0.65] {in($a_i$)} (si);
		\draw[action] (s.north east) -- (0.8,1.2) -- node[above,pos=0.7] {in($a_1$)} (s1);
		\node[dots] at (1.5,1) {$\vdots$};
		\draw[action] (s.south east) -- (0.8,-1.2) -- node[above,pos=0.7] {in($a_r$)} (sr);
		\node[dots] at (1.5,-.2) {$\vdots$};
		
		\node[scale=1.6] at (3.5,0) {$=$};
		\end{scope}
		
		\begin{scope}[xshift=2.2cm,yshift=-5.2cm,xscale=0.95,yscale=0.9]
		\tikzstyle{every node}=[ellipse, draw=black,
		inner sep=0pt, minimum width=10pt, minimum height=10pt]
		
		\draw (2,6)     node  (t1)	{};
		\draw (2.5,5)     node  (t2)	{};
		\draw (3.5,5)     node  (t3)	{};
		\draw (4,6)     node (t4)	{};
		\draw (3.5,7)     node [draw=none] (t5)	{...};
		\draw (2.5,7)     node (t6)	{};
		
		\draw (1,7)     node (s)	{$s$};
		\draw (1,6)     node (s1)	{$s_1$};
		\draw (1.75,4.25)     node (s2)	{$s_2$};
		\draw (4.25,4.25)     node  (s3)	{$s_3$};
		\draw (5,6)     node  (s4)	{$s_4$};
		%\draw (4.5,8)     node [draw=none] (s5)	{...};
		\draw (1.75,7.75)     node (s6)	{$s_r$};
		
		\path[->,action] (t1) edge node [above right, draw=none] {$\lambda$} (t2);
		\path[->,action] (t2) edge node [above, draw=none] {$\lambda$} (t3);
		\path[->,action] (t3) edge node [above left, draw=none] {$\lambda$} (t4);
		\path[->,action] (t4) edge node [below left, draw=none] {$\lambda$} (t5);
		\path[->,action] (t5) edge node [below, draw=none] {$\lambda$} (t6);
		\path[->,action] (t6) edge node [below right, draw=none] {$\lambda$} (t1);
		\path[->,action] (s) edge node [above right, draw=none] {$\lambda$} (t1);
		
		\path[->,action] (t1) edge node [below, draw=none] {$a_1$} (s1);
		\path[->,action] (t2) edge node [above left, draw=none] {$a_2$} (s2);
		\path[->,action] (t3) edge node [above right, draw=none] {$a_3$} (s3);
		\path[->,action] (t4) edge node [above, draw=none] {$a_4$} (s4);
		\path[->,action] (t6) edge node [below left, draw=none] {$a_r$} (s6);
		
		\end{scope}

		% % % % % % % % % % % % OUT
		
		\begin{scope}[xshift=7.8cm,xscale=0.75]
		
		\node (s) at (0,0) {$s$};
		
		\node[] (si) at (2.6,0) {$s_i$};
		\node[] (sr) at (2.6,-1.2) {$s_r$};
		\node[] (s1) at (2.6,1.2) {$s_1$};
		
		\path[action] (s) edge node[above,pos=0.65] {out($a_i$)} (si);
		\draw[action] (s.north east) -- (0.8,1.2) -- node[above,pos=0.7] {out($a_1$)} (s1);
		\node[dots] at (1.5,1) {$\vdots$};
		\draw[action] (s.south east) -- (0.8,-1.2) -- node[above,pos=0.7] {out($a_r$)} (sr);
		\node[dots] at (1.5,-.2) {$\vdots$};
		
		\node[scale=1.6] at (3.4,0) {$=$};
		
		\begin{scope}[xshift=4.4cm,xscale=0.8]
		
		\node[state] (s) at (0,0) {$s$};
		
		\node[state] (si) at (2.5,0) {$s_i$};
		\node[state] (sr) at (2.5,-1.2) {$s_r$};
		\node[state] (s1) at (2.5,1.2) {$s_1$};
		
		\path[action] (s) edge node[above,pos=0.65] {$a_i$} (si);
		\draw[action] (s.north east) -- (0.8,1.2) -- node[above,pos=0.7] {$a_1$} (s1);
		\node[dots] at (1.2,1) {$\vdots$};
		\draw[action] (s.south east) -- (0.8,-1.2) -- node[above,pos=0.7] {$a_r$} (sr);
		\node[dots] at (1.2,-.2) {$\vdots$};
		
		\end{scope}
		
		\end{scope}
		
		\end{scope}
		
		\end{tikzpicture}
	\end{center}
	\caption{Module for player $i$ on the left. Input and output encoding to the right.}
	\label{fig:syncmodule}
	
\end{figure}

%For the proof, 
We construct the distributed IMC so that only the outputting player chooses what action to output whereas the inputting player accepts whatever comes.
\mycut{
\begin{itemize}
	\item Outputting an action $a\in \{a_1,\ldots,a_r\}$ in a state $s$ is simple -- this is modelled in $s$ by standard action transitions for all these actions. 
	\item Inputting in a state $s$ one action from the set $\{a_1,\ldots,a_r\}$ is more elaborate. Instead of waiting in $s$, the player travels by delay transitions in a round-robin fashion through a cycle of $r$ states, where in the $i$-th state, only the action $a_i$ is available. This way, the player cannot influence anything and must input the action that comes.
\end{itemize}
}
The construction of modules for player $i$ is illustrated in Figure~\ref{fig:syncmodule} along with constructions for input and output. 
The interesting part is 
%the inputting mechanism. Inputting 
how an action from the set $\{a_1,\ldots,a_r\}$ is input in a state $s$.
%is done as follows. 
Instead of waiting in $s$, the player travels by delay transitions in a round-robin fashion through a cycle of $r$ states, where in the $i$-th state, only the action $a_i$ is available. Thus, the player has no influence and must input the action that comes. 
By this construction, the main module has at most one action transition in every state such that the player cannot influence anything; other modules  get no insight by observing time and thus the players have the same power as in the DEC-POMDP.
%. Thus, there is only one trivial strategy that cannot influence anything.
%
%For details, see~\citeapp{app:reduction}.
%
\qed
\end{proof}

\subsection{Undecidability of qualitative existence in DEC-POMDP}

Next, we show that the qualitative existence problem for DEC-POMDPs even with $n \geq 2$ players is undecidable. The proof has similarities with ideas from \cite{BK10} where it is shown that deciding existence of sure-winning strategies in safety games with 3 players and partial observation is undecidable. Using the randomness of DEC-POMDPs we show undecidability of the qualitative existence problem for reachability in $2$-player DEC-POMDPs. %We only give a proof sketch here, details can be found in the appendix.

\usetikzlibrary{shapes,backgrounds,calc}

\makeatletter
\tikzset{circle split part fill/.style  args={#1,#2}{%
		alias=tmp@name, % Jake's idea !!
		postaction={%
			insert path={
				\pgfextra{% 
					\pgfpointdiff{\pgfpointanchor{\pgf@node@name}{center}}%
					{\pgfpointanchor{\pgf@node@name}{east}}%            
					\pgfmathsetmacro\insiderad{\pgf@x}
					%\begin{scope}[on background layer]
					%\fill[#1] (\pgf@node@name.base) ([xshift=-\pgflinewidth]\pgf@node@name.east) arc
					%                    (0:180:\insiderad-0.5\pgflinewidth)--cycle;
					%\fill[#2] (\pgf@node@name.base) ([xshift=\pgflinewidth]\pgf@node@name.west)  arc
					%                     (180:360:\insiderad-0.5\pgflinewidth)--cycle;  
					\fill[#1] (\pgf@node@name.base) ([xshift=-\pgflinewidth]\pgf@node@name.east) arc
					(0:180:\insiderad-\pgflinewidth)--cycle;
					\fill[#2] (\pgf@node@name.base) ([xshift=\pgflinewidth]\pgf@node@name.west)  arc
					(180:360:\insiderad-\pgflinewidth)--cycle;            %  \end{scope}   
				}}}}}  
				\makeatother  
				
				\begin{figure}
					\begin{center}
						\vspace{-1em}
						\begin{tikzpicture}[xscale=0.9,yscale=0.7]
						
						\tikzstyle{every node}=[ellipse, draw=black,
						inner sep=0pt, minimum width=8pt, minimum height=8pt, outer sep=1pt]

						\draw (3,7)     node [label=below right:$s_0$] (s0)	{};
						\draw (3,6)     node [draw=none] (inv)	{};
						\draw (1,6)     node[shape=circle split,
						draw=black,
						line width=0mm,text=white,font=\bfseries,
						circle split part fill={black,black}, label=above left:$s_1$] (s1)	{};
						\draw (5,6)     node [label=below left:$s_2$] (s2)	{};
						\draw (8,7)     node [shape=circle split,
						draw=black,
						line width=0mm,text=white,font=\bfseries,
						circle split part fill={black,white},label=below right:$s_3$] (s3)	{};
						\draw (8,6)     node [shape=circle split,
						draw=black,
						line width=0mm,text=white,font=\bfseries,
						circle split part fill={black,black},label=below right:$s_4$] (s4)	{};

						%\draw (2,2)     node [draw=none] (m1)	{$M_1$};
						%\draw (4,2)     node [draw=none] (m2)	{$M_2$};
						\draw (-1,6)     node [draw=none] (m1)	{$\pomdp'$};
						\draw (10,6)     node [draw=none] (m2)	{$\pomdp''$};
						\draw (10,7)     node [draw=none] (m3)	{$\pomdp'''$};

						\path[-] (s0) edge node [left, draw=none,font=\footnotesize] {$(n,n)$} (3,6);
						\path[->] (3,6) edge node [below, draw=none] {$\frac{1}{3}$} (s1);
						\path[->] (3,6) edge node [below, draw=none] {$\frac{1}{3}$} (s2);
						\path[->] (3,6) edge [bend left = 15] node [above left, draw=none] {$\frac{1}{3}$} (s3);
						
						\path[-] (s2) edge [bend right] node [below, draw=none,font=\footnotesize] {$(n,n)$} (6,6);
						\path[->] (6,6) edge [bend right] node [above, draw=none] {$\frac{1}{3}$} (s2);
						\path[->] (6,6) edge node [above left, draw=none] {$\frac{1}{3}$} (s3);
						\path[->] (6,6) edge node [below left,pos=0.55, draw=none] {$\frac{1}{3}$} (s4);
						%\path[->] (4,4) edge node [above left, draw=none] {$\frac{1}{4}$} (sf);
						
						\path[->] (s3) edge (m3);
						\path[->] (s1) edge (m1);
						\path[->] (s4) edge (m2);

						\end{tikzpicture}
					\end{center}
					\caption{Overall structure of $\pomdp$ without details of $\pomdp', \pomdp''$ and $\pomdp'''$.}
					\label{fig:undec_overall}
				\end{figure}

\begin{mytheo} 
\label{theo:decpomdp_undec}
% The qualitative existence problem for 
It is undecidable whether for a DEC-POMDP $\pomdp$ with $n \geq 2$ players and a set $T$ of target states in $\pomdp$ if there exists a strategy profile $\vect{\sigma}$ such that $\pr^{\vect{\sigma}}_\pomdp(\diamond T) = 1$.
\end{mytheo}

\begin{proof}[Proof sketch] We do a reduction from the non-halting problem of a deterministic Turing machine $M$ that starts with a blank input tape. From $M$ we construct a DEC-POMDP $\pomdp$ with two players $\agt = \{1, 2\}$ such that $M$ does not halt if and only if players $1$ and $2$ have strategies $\vect{\sigma} = (\sigma_1,\sigma_2)$ which ensure that the probability of reaching a target set $T$ is 1. Figure \ref{fig:undec_overall} shows the overall structure of $\pomdp$ without details of sub-modules $\pomdp', \pomdp''$ and $\pomdp'''$.

Both players have two possible observations, black and white. 
We depict the observation of player 1 in the top-half and of player 2 in the bottom-half of every state.
%
%In Figure \ref{fig:undec_overall} the colour of the top-half of a state is the observation of player 1, the colour of the bottom-half of a state is the observation of player 2.
%
%In states $s_1,s_3$ and $s_4$ player 1 receives the black observation and in all other states he receive the white observation. Player 2 gets the black observation for states $s_1,s_4$ as well as in the initial state of $\pomdp'''$ . In all other states he receives another observation. 
%
The play starts in $s_0$ and with probability 1, every player receives the black observation exactly once during the play. If the play goes to $s_1$ or $s_4$ the players will receive the observation at the same time and if the play goes to $s_3$ then player 2 will receive the observation in the step after player 1 does. The modules $\pomdp', \pomdp''$ and $\pomdp'''$ are designed so that:

\begin{itemize}
 \item In $\pomdp'$, a target state is reached if and only if the sequence of actions played by both players encodes the initial configuration of $M$.
% , ending with a special symbol $\#$.
 
 \item In $\pomdp''$, a target state is reached with probability 1 if and only if both players play the same infinite sequence of actions. Note that randomness is essential to build such a module.
% , ending with a special symbol $\#$.
 
 \item In $\pomdp'''$, the target set is reached if and only if the sequences of actions played by player 1 and 2 encode two configurations $C_1$ and $C_2$ of $M$, respectively, such that $C_1$ is not an accepting configuration and $C_2$ is a successor configuration of $C_1$. This can be done since a finite automaton can be constructed that recognizes if one configuration is a successor of the other when reading both configurations at the same time. Note that it is possible because such configurations can only differ by a constant amount (the control state, the tape head position and in symbols in cells near the tape head).
\end{itemize}

It can be shown by induction that if there are strategies $\sigma_1, \sigma_2$ that ensure reaching $T$ with probability 1 then every
$\sigma_i$ has to play the encoding of the $j$th configuration of $M$ when it
%player $i$
% \in \{1,2\}$ 
receives the black observation in the $j$th step. Further, it can be shown that these strategies do ensure reaching $T$ with probability 1 if $M$ does not halt on the empty input tape and do not ensure reaching $T$ with probability 1 if $M$ halts. %For more details, see~\citeapp{app:undec_decpomdp}.
%two things must be satisfied 1) 
%the play leaves $s_0$. 
\qed
\end{proof}

%Since a 1-player DEC-POMDP is a  we also get the following.

%% file: results-reachability.tex
\section{Decidability for non-urgent models}\label{sec:results-reachability}

%
%
% Without loss of generality, we can assume that there are sets $T_1, \ldots, T_n$ such that $T = T_1 \times \cdots \times T_n$ and every local target state $t_i \in T_i$ is absorbing, i.e. has only one delay outgoing self-loop transition. Indeed, if it is not the case we can add a fresh state $t_i$ to each $S_i$

In this section, we turn our attention to a subclass of distributed IMCs, called \emph{non-urgent}, that implies decidability for both the qualitative and quantitative value problems for 2 players.

\begin{mydef}\label{def:redecide}
We call $\structure = ((S_i, \act_i, \inter{}_i, \marko{}_i, s_{0i}))_{1 \leq i \leq n}$ \emph{non-urgent} if for every $1 \leq j \leq n$:
\begin{enumerate}
	\item Every $s\in S_j$ is of one of the following forms:
	\begin{enumerate}
		\item \emph{Synchronisation} state with at least one outgoing synchronisation action transition and exactly one outgoing delay transition which is a self-loop.
		\item \emph{Private} state with arbitrary outgoing delay transitions and private action transitions.
	\end{enumerate}
	\item Player $j$ has an action $\donothing_j \in \act_j$ enabled in every synchronisation state from $S_j$ that allows to ``do nothing'' and thus postpone the synchronisation. To this end, $\donothing_j$ is also in $\act_k$ for every other player $k \neq j$ but $\donothing_j$ does not appear in $\inter{}_k$. As a result, $j$ does not take part in any synchronisation while  choosing $\donothing_j$.
\end{enumerate}
\end{mydef}

In a non-urgent distributed IMC, $\vect{s} \in S$ is called a (global) \emph{synchronisation} state if it is the initial state or all $\vect{s}(j)$ are synchronisation states. We denote global synchronisation states by $S'$. All other global states $S \setminus S'$ are called \emph{private}.

\begin{example}
  Consider the non-urgent variant of Example~\ref{ex1} on the right. The ``do nothing'' actions are a natural concept; the only real modelling restriction is that one cannot model a communication time-out any more, the delay transitions from synchronisation states need to be self-loops.
\end{example}
\vspace*{-2ex}

\begin{wrapfigure}[11]{r}{0.42\linewidth}
	%\begin{center}
	\vspace{-1.8em}
	\begin{tikzpicture}[yscale=0.45,xscale=0.45]
	\scriptsize
	
	\tikzstyle{every node}=[ellipse, draw=black,
	inner sep=0pt, minimum width=13pt, minimum height=13pt]
	
	\draw (0,8.3) node [draw=none] (app)	{\textbf{App:}};
	\draw (0,6) node (s0)	{$c_0$};
	\draw (2,7.3) node (l1)	{$t_1$};
	\draw (5,7.3) node (l2)	{$t_2$};
	\draw (8,7.3) node (l3)	{$t_3$};
	\draw (10,7.3) node (l4)	{$t_4$};
	\draw (2,4.7) node (r1)	{$b_1$};
	\draw (5,4.7) node (r2)	{$b_2$};
	\draw (8,4.7) node (r3)	{$b_3$};
	\draw (10,4.7) node (r4)	{$b_4$};
	
	\path[->] ($(s0)+(-1,0)$) edge (s0);
	
	\path[->] (s0) edge node [above left, draw=none] {$\lambda$} (l1);
	\path[->] (s0) edge node [below left, draw=none] {$\lambda$} (r1);
	
	\path[->] (l1) edge [loop above] node [right, draw=none] {$\lambda$} (l1);
	\path[->] (l2) edge [loop above] node [right, draw=none] {$\lambda$} (l2);
	\path[->] (l4) edge [loop above] node [left, draw=none] {$\lambda$} (l4);
	\path[->] (l1) edge [loop below] node [left, draw=none] {$\donothing_1$} (l1);
	\path[->] (l2) edge [loop below] node [left, draw=none] {$\donothing_1$} (l2);
	\path[->] (l1) edge node [above=-3, draw=none] {\login} (l2);	
	\path[->] (l2) edge node [above=-3, draw=none] {\lookup} (l3);
	%	\path[->] (l2) edge [bend right = 45] node [above=-3, draw=none] {\logout} (l1);
	\path[->] (l3) edge (l4);
	
	\path[->] (r1) edge [loop above] node [right, draw=none] {$\lambda$} (r1);
	\path[->] (r2) edge [loop above] node [right, draw=none] {$\lambda$} (r2);
	\path[->] (r4) edge [loop above] node [left, draw=none] {$\lambda$} (r4);
	\path[->] (r1) edge [loop below] node [left, draw=none] {$\donothing_1$} (r1);
	\path[->] (r2) edge [loop below] node [left, draw=none] {$\donothing_1$} (r2);
	\path[->] (r1) edge node [above=-3, draw=none] {\login} (r2);	
	\path[->] (r2) edge node [above=-3, draw=none] {\lookup} (r3);
	%	\path[->] (r2) edge [bend right = 45] node [above=-3, draw=none] {\logout} (r1);
	\path[->] (r3) edge (r4);
	
	\begin{scope}[xshift=-10cm,yshift=-4.8cm]
	\draw (10,7) node [draw=none] (app)	{\textbf{Att:}};
	
	\draw (12,6) node (t0)	{$\bar{c}_1$};
	\draw (15,6) node (t1)	{$\bar{c}_2$};
	\draw (18,6) node (t2)	{$\bar{c}_3$};
	\draw (20,7) node (t3)	{$\bar{t}_4$};
	\draw (20,5) node (t4)	{$\bar{b}_4$};

	\path[->] ($(t0)+(-1,0)$) edge (t0);
	
	\path[->] (t0) edge [loop above] node [right, draw=none] {$\lambda$} (t0);
	\path[->] (t1) edge [loop above] node [right, draw=none] {$\lambda$} (t1);
	\path[->] (t0) edge [loop below] node [left, draw=none] {$\donothing_2$} (t0);
	\path[->] (t1) edge [loop below] node [left, draw=none] {$\donothing_2$} (t1);
	
	\path[->] (t0) edge node [above=-3, draw=none] {\login} (t1);
	\path[->] (t1) edge node [above=-3, draw=none] {\lookup} (t2);
	\path[->] (t2) edge (t3);
	\path[->] (t2) edge (t4);
	\path[->] (t3) edge [loop above] node [left, draw=none] {$\lambda$} (t3);
	\path[->] (t4) edge [loop above] node [left, draw=none] {$\lambda$} (t4);
	\end{scope}
	\end{tikzpicture}
	%\end{center}
	%\caption{Transformation for states with synchronization transitions that adds new private actions $b_1, \ldots, b_n$.}
	%	\label{fig:nondeterminism}
\end{wrapfigure}

 Surprisingly, in this model, the secret can be leaked with probability $1$ as follows. As before, the players reach the states
 $(t_2,\bar{c}_2)$ or $(b_2,\bar{c}_2)$ with equal probability. Now, the $\app$ player can arbitrarily postpone the $\lookup$ by committing to action $\donothing_1$. Whenever the delay self-loop is taken, the player can re-decide to perform $\lookup$. Since the self-loop is taken repetitively, the $\app$ player is flexible in choosing the timing of $\lookup$. Thus, leaking the secret is simple, e.g. by performing $\lookup$ in an odd second when in $t_2$ and in an even second when in $b_2$. 
 \vspace*{1ex}

For two players, we construct a general synchronisation scheme that (highly probably) shows the players the current global state after each communication.

\begin{mytheo}\label{thm:alg}
The quantitative value problem for 2-player non-urgent distributed IMCs where the target set consists only of synchronisation states is in {\exptime}.
%There is an {\exptime} algorithm for the quantitative value problem for 2-player non-urgent distributed IMCs where $T$ is a subset of synchronization states.
\end{mytheo}

Being a special case, also the qualitative value problem is decidable. 
%Note that also the qualitative problem is decidable as it is a special case of the quantitative problem.
%
In essence, the problem becomes decidable because in the synchronisation states, the players can effectively exchange arbitrary information. This resembles the setting of~\cite{DBLP:conf/icalp/GenestGMW13}.
The insight that observing global time provides an additional
synchronization mechanism is not novel in itself, but it is obviously
burdensome to formally capture in time-abstract models of asynchronous
communication, and thus usually not considered.  For distributed IMC, it
still is non-trivial to develop; here it hinges on the non-urgency
assumption.
The results of~\cite{DBLP:conf/icalp/GenestGMW13} also indicate that for three or more players, additional constraints on the topology may be needed to obtain decidability.

%It is nothing novel that observing global time gives additional synchronization mechanism. Still this mechanism (a) is usually not addressed in discrete-time abstractions of asynchronous communication models and (b) is also not trivial to employ in the setting of distributed IMC (and we need to make the non-urgency assumption). 
%Such additional communication channels inevitably arise in systems with asynchronous communication and were surprisingly ignored in the discrete time models.
%To the best of our knowledge, were surprisingly never considered.

In the rest of the section we prove Theorem~\ref{thm:alg},
fixing a 2-player non-urgent distributed IMC $\img = ((S_i, \act_i, \inter{}_i, \marko{}_i, s_{0i}))_{1 \leq i \leq 2}$, $p \in [0,1]$, and $T \subseteq S'$.
We present the algorithm based on a reduction to a discrete-time Markov decision process and then discuss its correctness.

\paragraph{Markov decision process (MDP)} An MDP is a tuple $\mdp = (\states, \actions, \tra, s_0)$ where $\states$ is a finite set of states,
$\actions$ is a finite set of actions, $\tra:\states\times\actions\rightarrow \dist(S)$
is a partial probabilistic transition function, and $s_0$ is an initial state. An MDP is the special case of a DEC-POMDP with 1 player that has a unique observation for each state.
A {\em play} in $\mdp$ is a sequence $\omega = s_0 s_1 \ldots$ of states such that $\tra(s_i,a_i)(s_{i+1}) > 0$ for some action $a_i$ for every $i \geq 0$. A history is a prefix of a play. 
A {\em strategy} is a function $\pi$ that to every history $h\cdot s$ assigns a probability distribution over actions such that
if an action $a$ is assigned a non-zero probability, then $\tra(s,a)$ is defined. A strategy $\pi$ is {\em pure memoryless} if it assigns
Dirac distributions to any history and its choice depends only on the last state of the history.
When we fix a strategy $\pi$, we obtain a 
probability measure $\pr^{\pi}$ over the set of plays. For further details, see~\cite{puterman2009markov}.

\paragraph{The algorithm} It works by reduction to an MDP $\mdp_\structure = (S',A,P,\vect{s_0})$ where 
\begin{itemize}
	\item $S' \subseteq S$ is the set of global synchronisation states;
	\item $A = \choice \times \Sigma_1 \times \Sigma_2 \cup \{\bot\}$ where $\Sigma_j$ is the set of pure memoryless strategies of player $j$ in $\structure$ that choose $\donothing_j$ in every synchronisation state;
	\item For an arbitrary state $(s_1,s_2)$, we define the transition function as follows:
	\begin{itemize}
		\item
For any $(\vect{c},\sigma_1,\sigma_2) \in A$, the transition $P((s_1,s_2), (\vect{c},\sigma_1,\sigma_2))$ is defined if $\vect{c}$ is available in $(s_1,s_2)$ and the players agree in $\vect{c}$ on some action $a$, i.e. $\enab(\vect{c}) = \{a\}$. If defined, the distribution $P((s_1,s_2), (\vect{c},\sigma_1,\sigma_2))$ assigns to any successor state $(s'_1,s'_2)$ the probability that in $\img$ the state $(s'_1,s'_2)$ is reached from $(s_1,s_2)$ via states in $S \setminus S'$ by choosing $\vect{c}$ and then using the pure memoryless strategy profile $(\sigma_1,\sigma_2)$.
		\item To avoid deadlocks, the transition $P((s_1,s_2), \bot)$ is defined iff no other transition is defined in $(s_1,s_2)$ and it is a self-loop, i.e. it assigns probability $1$ to $(s_1,s_2)$.
	\end{itemize}
\end{itemize}

The MDP $\mdp_\structure$ has size exponential in $|\structure|$. 
Note that all target states $T$ are included in $S'$. Slightly abusing notation, let $\diamond T$ denote the set of plays in $\mdp_\structure$ that reach the set $T$.
From standard results on MDPs~\cite{puterman2009markov}, there exists an optimal pure memoryless strategy $\pi^\ast$, i.e. a strategy satisfying
%that 
%from any state $(s_1,s_2) \in S'$ 
%maximizes  
$\pr^{\pi^\ast} (\diamond T) = \sup_\pi \pr^{\pi}
%_{(s_1,s_2)}
(\diamond T)$. 
Furthermore, such a strategy $\pi^\ast$ and the value $v := \pr^{\pi^\ast} (\diamond T)$ can be computed in time polynomial in $|\mdp_\structure|$. Finally, the algorithm returns TRUE if $v \geq p$ and FALSE, otherwise.

\paragraph{Correctness of the algorithm} Let us explain why the approach is correct.

\begin{myprop}\label{prop:correct}
	The value of $\structure$ is equal to the value of $\mdp_\structure$, i.e.
	$$
	\sup_{\vect{\sigma}} \pr^{\vect{\sigma}}(\diamond T) 
	\;\; = \;\; 
	\sup_\pi \pr^\pi(\diamond T).
	$$
\end{myprop}

\paragraph{Proof sketch.}
%	As the first step in the proof, we transfer the problem to $\structure$ by defining \emph{global synchronisation} strategies $\theta$ in $\structure$ that have the same power as the strategies in $\mdp_\structure$, i.e.
%	$$
%	\sup_{\theta} \pr^{\theta}(\diamond T) 
%	\;\; = \;\; 
%	\sup_\pi \pr^\pi(\diamond T).
%	$$
	As regards the $\leq$ inequality, it suffices to show that any strategy profile $\vect{\sigma}$ can be mimicked by some strategy $\pi$. This is simple as $\pi$ in $\mdp_\structure$ has always knowledge of the global state.
	Much more interesting is the $\geq$ inequality. We argue that for any strategy $\pi$ there is a sequence of local strategy profiles ${\vect{\sigma}}^1, {\vect{\sigma}}^2, \ldots$ such that 
	$$
	\lim_{i\to\infty} \pr^{{\vect{\sigma}}^i}(\diamond T)
	= 
	\pr^{\pi}(\diamond T).$$

	\begin{wrapfigure}[10]{r}{0.42\linewidth}
			\vspace{-4.5em}
		\begin{tikzpicture}[xscale=1,yscale=1,font=\scriptsize]
		
		\draw[pattern=north west lines, pattern color=black!30] (0,0) rectangle (1,1);
		\draw[pattern=north west lines, pattern color=black!30] (2,0) rectangle (3,1);
		
		\draw[pattern=north east lines, pattern color=black!30] (0,0) rectangle (1,3);
		
		\draw[thick, scale=1] (0, 0) grid (4, 3);
		
		\draw [decoration={brace}, decorate] (-.7,0) -- (-.7,3);
		
		\draw [decoration={brace}, decorate] (0,3.7) -- (4,3.7);
		
		\node[]  at (2,4) {$S'_2$};
		\node[draw,circle]  at (1,3.4) {$1$};
		\node[]  at (3,3.4) {$2$};
		
		\node[]  at (.5,-.3) {$\sync$};
		\node[]  at (1.5,-.3) {$\neg \sync$};
		\node[]  at (2.5,-.3) {$\sync$};
		\node[]  at (3.5,-.3) {$\neg \sync$};
		
		\node[]  at (-1,1.5) {$S'_1$};
		\node[]  at (-.4,2.5) {$1$};
		\node[]  at (-.4,1.5) {$2$};
		\node[draw,circle]  at (-.4,.5) {$3$};

		\node[]  at (.5,2.7) {$\donothing_1$};
		\node[]  at (.5,2.3) {$c_2(1,1)$};
		
		\node[]  at (.5,1.7) {$\donothing_1$};
		\node[]  at (.5,1.3) {$c_2(2,1)$};
		
		\node[]  at (.5,.7) {$c_1(3,1)$};
		\node[]  at (.5,.3) {$c_2(3,1)$};
		
		\node[]  at (2.5,.7) {$c_1(3,2)$};
		\node[]  at (2.5,.3) {$\donothing_2$};
		
		\node[]  at (2.5,1.7) {$\donothing_1$};
		\node[]  at (2.5,1.3) {$\donothing_2$};
		\node[]  at (2.5,2.7) {$\donothing_1$};
		\node[]  at (2.5,2.3) {$\donothing_2$};

		\end{tikzpicture}
%		\caption{Blabla}
		\label{fig:grid}
	\end{wrapfigure}

The crucial idea is that a strategy profile communicates correctly (with high probability) the current global state in a synchronisation state by \emph{delaying} as follows. The time is divided into phases, each of $|S'_1|\cdot 2|S'_2|$ slots (where $S'_i$ are the synchronisation states of player $i$). We depict a phase by the table on the right where the time flows from top to bottom and from left to right (as reading a book).
			Players $1$ and $2$ try to synchronise in the row and column, respectively, corresponding to their current states (in circle) and in each slot take the choice $c_i(s_1,s_2)$ optimal given the current global state is $(s_1,s_2)$; in the remaining slots they choose to do nothing.
			Since the players can change their choice only at random moments of time, their synchronising choice always stretches a bit into the successive silent slot (in a $\neg \sync$ column).
			The more we increase the size of each slot, the lower is the chance that a synchronisation choice of a player stretches to the next \emph{synchronisation} slot. Thus, the lower is the chance of an erroneous synchronisation.
			We define the size of the slot to increase with $i$ and also along the play so that for any fixed $i$ the probability of at least one erroneous synchronisation is bounded by $\kappa_i < 1$ and for $i\to \infty$, we have $\kappa_i \to 0$.
\qed

%% file: discussion.tex
\section{Discussion and Conclusion}
\label{sec:discussion}

This paper has introduced a foundational framework for modelling and
synthesising distributed controllers interacting in continuous time
via handshake communication. The continuous time nature of the model
induces that the interleaving scheduler has in fact very little
power. We studied cooperative reachability problems for which we
presented a number of undecidability results, while for non-urgent
models we established decidability of both quantitative and
qualitative problems for the two-player case. In the framework
considered, the restriction to exponential distributions is a
technical vehicle, the results could have been developed in general
continuous time, e.g.~by using the model of stochastic
automata~\cite{dargenio-sa}.

Distributed IMCs can be considered as an attractive base model
especially in the context of information flow and other
security-related studies. This is because in contrast to the discrete
time setting, the power of the interleaving scheduler is no matter of
debate, it can leak no essential information. 

From a more general perspective, distributed synthesis of control
algorithms has received considerable attention in its
entirety~\cite{pnueli-rosner,pnueli-rosner-async,madhu1}. The asynchronous setting with handshake
synchronisation has been considered in~\cite{madhu}. Notably, our
assumption that players stay committed to a particular action choice
for the time in between state changes implies the necessity to let the
players explicitly solve distributed consensus problems. As done
in~\cite{madhu}, one can overcome this by letting local
players pick sets of enabled actions (or letting them change choice with infinite speed), and then let some built-in magic
pick a valid action from the intersection, 
%in favour of or against the winning condition,
implying that whenever possible a consensus is
reached for sure. Such a change would however reintroduce the scheduler.
      
We should point out that distributed IMCs are not fully
compositional: We are assuming a fixed vector of modules, and do not
discuss that modules themselves may be vectors. 
Otherwise, we would face the
phenomenon of auto-concurrency~\cite{auto}, where transitions with
identical synchronisation actions might get enabled concurrently,
despite not synchronising. This in turn would again re-introduce distinguishing power of the scheduler. 
%In our setting, all relevant choices are either random choices that stem from delay transitions or are induced by strictly local players. % The absence of full
%compositionality is one of the main differences between distributed
%structures and the standard process-algebraic way of supporting
%parallel composition with handshake~\cite{CSP,lotos,pepa,prism}.  The
%closest model extending distributed structures in this respect is the
%IMC algebra~\cite{Her02}.
%Markovski~\cite{Markovski-ICCA,Markovski-ACSD} has worked on
%supervisory control with IMC, which however lacks the distributed game
%nature of our problem. 

Distributed Markov chains~\cite{SEJMT15}
constitute another recent discrete-time approach where interleaving
nondeterminism is tamed successfully via assumptions on the
communication structure.
The observation that continuous time reduces the power of the
interleaving scheduler is not entirely new. Though not explicitly
discussed, it underpins the model of probabilistic I/O automata
(PIOA)~\cite{WSS97} which uses I/O communication with
input-enabledness and output-determinism. In that setting,
output-determinism implies that local players have no decisive power,
and hence a continuous time Markov chain arises. 
We can approximate
I/O-based communication by distributed IMCs without the need
for output-determinism. The approximation is linked to arbitrarily small
but non-zero delays needed to cycle through synchronising action
sets. A profound investigation of the continuous-time particularities
of this and other synchronisation disciplines is considered an
interesting topic for future work.

\begin{paragraph}{Acknowledgements}
	This work is supported
	by the EU 7th Framework Programme projects 295261
	(MEALS) and 318490 (SENSATION), by the Czech Science Foundation project P202/12/G061, the DFG Transregional Collaborative
	Research Centre SFB/TR 14 AVACS, and by the CDZ project 1023 (CAP).
	%  , and by the CAS/SAFEA International  Partnership Program for Creative Research Teams.
\end{paragraph}

%% file: appendix.tex
\section{Definition of the probability measure} \label{app:defs}

\begin{figure}[b]
	\begin{center}
	\begin{tikzpicture}[yscale=0.8,xscale=0.65]
	
	\tikzstyle{every node}=[ellipse, draw=black,
	inner sep=0pt, minimum width=10pt, minimum height=10pt]
	
	\draw (2,6) node [label=above:$s_0$] (s0)	{};
	\draw (1,5) node [label=below:$s_1$] (s1)	{};
	\draw (2,5) node [draw = none] (s2)	{...};
	\draw (3,5) node [label=below:$s_n$] (s3)	{};
	
	\draw (7,6) node [label=above:$s_0$] (t0)	{};
	\draw (6,4) node [label=below:$s_1$] (t1)	{};
	\draw (7,4) node [draw = none] (t2)	{...};
	\draw (8,4) node [label=below:$s_n$] (t3)	{};
	
	\draw (6,5) node (n1)	{};
	\draw (7,5) node [draw = none] (n2)	{...};
	\draw (8,5) node (n3)	{};
	
	\path[->] (3.5,5.5) edge [decorate, decoration={snake}] (5,5.5);
	\path[->] (s0) edge node [above left, draw=none] {$a_1$} (s1);
	\path[->] (s0) edge node [above right, draw=none] {$a_n$} (s3);
	
	\path[->] (t0) edge node [above left, draw=none] {$b_1$} (n1);
	\path[->] (t0) edge node [above right, draw=none] {$b_n$} (n3);
	
	\path[->] (n1) edge node [left, draw=none] {$a_1$} (t1);
	\path[->] (n3) edge node [right, draw=none] {$a_n$} (t3);
	
	\path[->] (n1) edge [loop left] node [left, draw=none] {$\lambda$} (n1);
	\path[->] (n3) edge [loop right] node [right, draw=none] {$\lambda$} (n3);
	
	\end{tikzpicture}
\end{center}
	\caption{Transformation for states with synchronisation transitions that adds new private actions $b_1, \ldots, b_n$.}
	\label{fig:transform}
\end{figure}
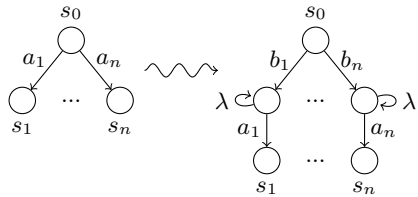
%\vspace{-1em}
First, we pose an assumption on the module of each player $j$ to avoid that a play stops because no further steps are possible.

\begin{myassumption}\label{ass:non-stop}
We assume that every state has at least one outgoing delay transition or it only has outgoing private action transitions (that cannot be blocked by other modules). 
\end{myassumption}
This is no real restriction: for states without any transition, we can add a delay self-loop; states with synchronisation transitions are transformed as depicted in Figure~\ref{fig:transform}.
Note that in the second case, simply adding delay self-loops would change the behaviour because taking a delay self-loop allows the player to change the choice.

As common for continuous-time systems, the definition is based on \emph{cylinder} sets generated from interval-timed histories. An \emph{interval-timed history} is a sequence 
$$H = \vect{s_0} \vect{c_0} \arrow{a_1,I_1}{} \vect{s_1} \vect{c_1}  \cdots \arrow{a_k,I_k}{} \vect{s_k}$$
where each $I_{i+1}$ is a real interval bounding the time spent \emph{waiting in $\vect{s_i}$}. We further require that $I_i = [0,0]$ whenever $a_i \in \act$ and that the set of histories that \emph{conform} to $H$ is non-empty. We say that a history 
$\vect{s_0} \vect{c_0} \arrow{a_1,t_1}{} \vect{s_1} \vect{c_1}  \cdots \arrow{a_k,t_k}{} \vect{s_k}$ \emph{conforms to $H$} if $t_i - t_{i-1} \in I_i$ for each $1 \leq i \leq k$ (where $t_0 = 0$).

%Note that we cannot explicitly restrict the times when synchronization $a_j$ occurs. This is nevertheless directly restricted by the interval $I$ of previous delay transition since the synchronization occurs without any delay after the delay transition.
%
Slightly abusing notation, we interpret an interval-timed history $H$ also as the set of histories that conform to $H$.
We define the cylinder $\cyl(H) = \{\rho \in \play \mid \exists i . \rho_{\leq i} \in H\}$ as the set of plays which have a prefix in $H$.
We define the measurable spaces $(\play, \sigmafield)$ and $(\hist, \sigmafielda)$ where $\sigmafield$ and $\sigmafielda$ are the $\sigma$-algebra generated from all cylinders and interval-timed histories, respectively:
\begin{align*}
\sigmafield &= \sigma\left( \{ \cyl(H) \mid \text{$H$ is an interval-timed history} \} \right) \\
\sigmafielda &= \sigma\left( \{ H \mid \text{$H$ is an interval-timed history} \} \right)
\end{align*}
Analogously, we can define measurable spaces $(\hist_j, \sigmafielda_j)$ over local histories (by allowing also $I_i$ for action transitions to have non-zero length).

For a given \emph{strategy profile} $\vect{\sigma}$, i.e. a tuple of strategies $\vect{\sigma} = (\sigma_1,\ldots,\sigma_n)$ of individual players, a given scheduler $\delta$ and initial state $\vect{s_0}$, we obtain a purely stochastic process and we can define a probability measure $\pr_{\vect{s_0}}^{\vect{\sigma},\delta}$ over $\play$. 
The probability measure is uniquely determined by fixing probabilities for cylinder sets $\cyl(H)$ for any interval-timed history
$$H = \vect{s_0} \vect{c_0} \arrow{a_1,I_1}{} \vect{s_1} \vect{c_1}  \cdots \arrow{a_k,I_k}{} \vect{s_k}.$$
Let $\alpha_1,...,\alpha_v$ be the indices of delay transitions, i.e. for each $\alpha_j$ we have $a_{\alpha_j} \in \agt$; and 
$\beta_1,...,\beta_u$ be the indices of action transitions, i.e. for each $\beta_j$ we have $a_{\beta_j} \in \act$.
For $1 \leq i \leq v$, let $\ell_i = \inf(I_{\alpha_i})$ and $u_i = \sup(I_{\alpha_i})$. 
Then $\pr_{\vect{s_0}}^{\vect{\sigma},\delta}(\cyl(H))$ is defined as
$$\displaystyle
\int_{\ell_1}^{u_1}... \int_{\ell_v}^{u_v} \prod_{1 \leq i \leq v} De^{\alpha_i}(d_i) \cdot \prod_{0 \leq i < k} St^{i} \cdot \prod_{1 \leq i \leq u} Sc^{\beta_i} \; \mathrm{d} d_v \cdots \mathrm{d} d_1$$
where the terms $De^i(d)$, $Sc^i$, and $St^i$ express the contribution of the $i$th transition to the overall probability caused by the delays, decisions of the scheduler, and decision on the strategies, respectively (see below).
The variables $d_i$ denote the delay at $i$th delay transitions and induce a history 
$$h = \vect{s_0} \vect{c_0} \arrow{a_1,t_1}{} \vect{s_1} \vect{c_1}  \cdots \arrow{a_k,t_k}{} \vect{s_k}$$ such that $t_i = \sum_{\ell,\alpha_\ell \leq i} d_{\alpha_\ell}$. 
Finally, we set
\begin{align*}
De^i(d) &= Q(\vect{s}_i,\vect{s}_{i+1})\cdot e^{-E(\vect{s}_i) d}, \\
\intertext{where $Q(\vect{s},\vect{s'}) = \sum_{\vect{s} \marko{\lambda, j} \vect{s'}} \lambda$ and $E(\vect{s}) = \sum_{\vect{s}' \neq \vect{s}} Q(\vect{s},\vect{s}')$, and}
St^i &= \prod_{j \in \syn(a_i)} \sigma_j(\pi_j(h_{\leq i}))(\vect{c_i}(j)), \\
Sc^i &= \delta(h_{\leq i-1},\vect{c_{i-1}})(a_i).
\end{align*}

When $\vect{s'_0} \neq \vect{s_0}$ then $\pr_{\vect{s'_0}}^{\vect{\sigma},\delta}(\cyl_\mathcal{G}(H)) = 0$.

\section{Proof of Theorem \ref{theo:sched_power}}

Before we prove Theorem \ref{theo:sched_power} we need a few definitions and lemmas.

For an interval-timed history let $\pi_j(H) = \{\pi_j(\rho) \in \play \mid \rho \textup{ conforms to } H \}$. Further, Let $\sim$ be an equivalence relation on interval-timed histories defined such that $H \sim H'$ for two interval-timed histories $H$ and $H'$ if and only if $\pi_j(H) = \pi_j(H')$ for every $j \in \agt$. We write $[H]$ for the set of histories $h$ such that there exists $H' \sim H$ with $h \in H'$. 

Now, for an interval-timed history $H = \vect{s_0} \vect{c_0} \arrow{a_1,I_1}{} \vect{s_1} \vect{c_1}  \cdots \arrow{a_k,I_k}{} \vect{s_k}$ such that the last action $a_k \in \agt$ if $k > 0$ we define the interleaving-abstract cylinder as the set of plays with a prefix conforming to interval-timed histories $H'$ assuring $H \sim H'$:
$$ \cylta(H) =\{ \rho \in \play \mid \rho_{\leq k} \in [H]\}$$
Note that the interleaving-abstract cylinders are contained in the $\sigma$-algebra $\sigmafield$ generated by the cylinder sets. We can therefore define a sub-$\sigma$-algebra $\sigmafieldb$ of $\sigmafield$ generated by interleaving-abstract cylinders
\begin{align*}
\sigmafieldb = \sigma(\{\cylta(H) \mid & H \textup{ is an interval-timed history }\\
 & \textup{s.t. the last action is not in } \act \})
\end{align*}
Since this is a sub-$\sigma$-algebra of $\sigmafield$ it inherits the probability measures defined earlier restricted to $\sigmafieldb$. Note that interleaving-abstract events are also 0-time abstract. That is, they are events which are invariant under reordering of  0-time interactions. Indeed, if no player can distinguish two histories $h$ and $h'$, then the delay transitions must be the same in the two histories. Further, since all players have access to global time it must be the same actions that are performed in $h$ and $h'$ in every 0-duration sub-history Now, the only way in which $h$ and $h'$ can differ is in the interleaving of these 0-duration subhistories. 

\begin{mylemma}
\label{lem:pureunique}
Let $\sigma = (\sigma_1,...,\sigma_n)$ be a pure strategy profile and let $\delta, \delta'$ be two schedulers. Further, let $H = \vect{s_0} \vect{c_0} \arrow{a_1,I_1}{} \vect{s_1} \vect{c_1}  \cdots \arrow{a_k,I_k}{} \vect{s_k}$ be an interval-timed history such that $a_i \in \act$ for all $1 \leq i < k$ and $a_k \in \agt$. Then
$$\pr_{\vect{s_0}}^{\sigma, \delta}(\cylta(H)) = \pr_{\vect{s_0}}^{\sigma, \delta'}(\cylta(H))$$
Further, this probability is either 0 or 1.
\end{mylemma}

\begin{proof}
First note that before a delay transition happens there can be no loops in any of the local histories $\pi_j(H)$. When we restrict to histories up to the first delay transition, a pure strategy for player $j$ can be considered to be simply a maximal sequence $(a^j_1, s^j_1) ... (a^j_{\ell_j}, s^j_{\ell_j})$ of choices of player $j$. This is because for any strategy $\sigma_j$ of player $j$, the $j$-play moves along the unique path in the module of player $j$ that is chosen by player $j$ at each step. The only thing player $j$ can observe is whether the transition (that he chose) is taken or not.

Assume now that we have a pure strategy profile $\sigma = (\sigma_1,...,\sigma_n)$. This induces such a sequence of choices for each player. The longest possible $j$-play that can occur under $\sigma$ before the first delay transition is $\rho^j = s_0(j) (a^j_1, s^j_1), \arrow{a^j_1, 0}{} s^j_1 ... \arrow{a^j_{\ell_j}, 0}{} s^j_{\ell_j}$. Further, all possible $j$-plays are prefixes of $\rho^j$. We now suppose for contradiction that there are two histories $h, h'$ consistent with $\sigma$ ending after the first delay transition such that there exists a player $j$ with $|\pi_j(h)| \neq |\pi_j(h')|$. Since $\pi_{j}(h)$ and $\pi_{j}(h')$ are both prefixes of $\rho^{j}$ either $\pi_{j}(h)$ is a proper prefix of $\pi_{j}(h')$ or the other way around. Without loss of generality let $u = |\pi_{j}(h)| < |\pi_{j}(h')|$.

Suppose $\sigma_j(\pi_j(h)) = (b_0, s)$. Now, player $j$ has not been able to synchronize on $b_0 \in \act$ in the last state of $h$. However, he has been able to synchronize on it along $h'$ due to different choices of the scheduler. Note that at earlier points on $h$ he might have synchronized on $b_0$ a number of times already. Let this number of times be $c_0 \in \mathbb{N}$. This means that all players capable of synchronizing on $b_0$ must have done so exactly $c_0$ times along $h$. At least one of these players, let us call him $j_1$, must have stopped before committing to synchronize on $b_0$ the $(c_0+1)$th time, because otherwise the play would have progressed since no action can be enabled when a delay transition takes place. Note that each of these players are willing to perform $b_0$ at least $c_0+1$ times at some point since this happens in $h'$. Let $j_1$ be committed to synchronize on action $b_1 \neq b_0$ when the play stops in $h$ and suppose he has already synchronized $c_1$ times on $b_1$ before this point. We can now perform the same reasoning again to find a player $j_2$ that has stopped in $h$ before reaching the point where he is ready to synchronize on $b_1$ for the $(c_1+1)$th time. At this point he is committed to performing the action $b_2$ for the $(c_2+1)$th time. This reasoning gives us an infinite sequence $j_i$ of players committed to actions $b_i$ in the last state after having performed the action $b_i$ exactly $c_i$ times before.

%If there are $i,i'$ such that $j_i = j_{i'}$ and either $c_i \neq c_{i'}$ or $b_i \neq b_{i'}$ we have a contradiction. 

We now introduce a partial order $\preceq \subseteq (\act \times \mathbb{N})^2$ on elements $(a,d) \in (\act \times \mathbb{N})$ such that there exists $j$ so action $a$ occurs at least $d$ times in $h'$. The relation is defined such that $(a,d) \preceq (a',d')$ if there is a player $j$ such that $a$ occurs $d$ times on $h'$ before $a'$ occurs $d'$ times on $h'$ (and $a'$ actually does occur $d'$ times at some point on $h'$). It is reflexive and transitive because all players that has an action in their alphabet must commit in order to synchronize on it. Anti-symmetry follows from this and the fact that $\rho^j$ is linear for every $j$.

Now, we have that $(b_{i+1},c_{i+1}) \prec (b_i,c_i)$ for all $i \ge 0$. This is the case firstly because $\pi_{j_i}(h)$ is a prefix of $\pi_{j_i}(h')$ for the players giving rise to the sequence of $b_i$'s. Secondly, because $b_{i+1} \neq b_i$. This means that

$$(b_0,c_0) \succ (b_1,c_1) \succ (b_2,c_2) ...$$

is an infinite strictly decreasing sequence. Since $\succ$ is only defined on a finite number of elements this gives a contradiction. Thus, $|\pi_j(h)| = |\pi_j(h')|$. Further, $\pi_j(h) = \pi_j(h')$ since one is a prefix of the other. Since $j$ was chosen arbitrarily, this means that the local history $\pi_j(h)$ before the delay transition cannot be changed by any scheduler and is uniquely determined by the pure strategies. From this, the lemma follows.
\qed
\end{proof}

We now extend to arbitrary events in $\sigmafieldb$ by applying the result above. However, first we need some notation. If $H = \vect{s_0} \vect{c_0} \arrow{a_1,I_1}{} \vect{s_1} \vect{c_1}  \cdots \arrow{a_k,I_k}{} \vect{s_k}$ is an interval-timed history then the set of histories $h \in [H]$ with specific delays $d_1,...,d_v$ on the delay transitions is denoted $[H]^{d_1,...,d_v}$. Note that this set is finite.

\begin{mylemma}
\label{lem:purestrats}
Let $E \in \sigmafieldb$, $\vect{s_0} \in S$, $\sigma = (\sigma_1,...,\sigma_n)$ be a pure strategy profile and $\delta,\delta'$ be two schedulers. Then

$$\pr_{\vect{s_0}}^{\sigma, \delta}(E) = \pr_{\vect{s_0}}^{\sigma, \delta'}(E)$$

\end{mylemma}

\begin{proof}
We show this by showing that for every time-abstract cylinder $\cylta(H)$ for an interval-timed history $H = \vect{s_0} \vect{c_0} \arrow{a_1,I_1}{} \vect{s_1} \vect{c_1}  \cdots \arrow{a_k,I_k}{} \vect{s_k}$ such that $a_k \in \agt$ if $k > 0$, every pure strategy profile $\sigma$ and every pair $\delta, \delta'$ of schedulers we have
 
$$\pr_{\vect{s_0}}^{\sigma, \delta}(\cylta(H)) = \pr_{\vect{s_0}}^{\sigma, \delta'}(\cylta(H))$$

First, for any scheduler $\delta$ we have

$$\pr_{\vect{s_0}}^{\sigma, \delta}(\cylta(H)) = \sum_{H' \sim H} \pr_{\vect{s_0}}^{\sigma, \delta}(\cyl(H'))$$

since these cylinder sets are disjoint.

Suppose	that $a_i \in \agt$ for the indices $\alpha_1,...,\alpha_v$ and $a_i \in \act$ for the indices $\beta_1,...,\beta_u$. For $1 \leq i \leq v$, let $\ell_i = \inf(I_{\alpha_i})$ and $u_i = \sup(I_{\alpha_i})$. Next, consider fixed delays $d_1,...,d_v$ and let $[H]^{d_1,...,d_v} = \{h^1,...,h^r\}$. For $1 \leq m \leq r$ we denote

$$h^m = \vect{s^m_0} \vect{c^m_0} \arrow{a^m_1,t_1}{} \vect{s^m_1} \vect{c^m_1}  \cdots \arrow{a^m_k,t_k}{} \vect{s^m_k}$$

where the timestamps $t_i$ are induced by the delays $d_1,...,d_v$ on the delay transitions. Now, since every interval-timed history $H' \sim H$ contains the same intervals and same delay transitions we have
\\

$\displaystyle \pr_{\vect{s_0}}^{\sigma, \delta}(\cylta(H))$\\

$\displaystyle = \sum_{H' \sim H} \pr_{\vect{s_0}}^{\sigma, \delta}(\cyl(H'))$\\

$\displaystyle = \int_{\ell_1}^{u_1}... \int_{\ell_v}^{u_v} \prod_{1 \leq i \leq v} Q(s^m_{\alpha_i}, s^m_{\alpha_i + 1}) \cdot e^{-E(d_i)}$\\

$\displaystyle\cdot \sum_{h^m \in [H]^{d_1,...,d_v}} \bigg ( \prod_{1 \leq i \leq u} \delta(h^m_{\leq \beta_{i} - 1}, c^m_{\beta_i - 1})(a^m_{\beta_i})$\\

$\displaystyle\cdot \prod_{0 \leq i < k} \prod_{j \in \syn(a^m_i)} \sigma_j(\pi_j(h^m_{\leq i}))(c^m_i(j)) \bigg)  \; \mathrm{d} d_v \cdots \mathrm{d} d_1$ 
\\

We will now show by induction on $v$ that for any fixed delays $d_1,...,d_v$
\\

$p^{d_1,...,d_v} = \displaystyle \sum_{h^m \in [H]^{d_1,...,d_v}} \bigg ( \prod_{1 \leq i \leq u} \delta(h^m_{\leq \beta_{i} - 1}, c^m_{\beta_i - 1})(a^m_{\beta_i})$
\\

$\displaystyle \prod_{0 \leq i < k} \prod_{j \in \syn(a^m_i)} \sigma_j(\pi_j(h^m_{\leq i}))(c^m_i(j)) \bigg)$
\\

either equals $0$ or $1$ independently of the scheduler $\delta$. This implies that $\displaystyle \pr_{\vect{s_0}}^{\sigma, \delta}(\cylta(H))$ is independent of the scheduler and thus proves the lemma.

For the base case suppose that $v = 0$. Then the only possibility is $H = \vect{s_0}$ since $H$ cannot end with an action transition. In this case it is immediate that $p^{d_1,...,d_v} = 1$.

For the inductive case suppose that $v > 0$ and that it holds for all $H'$ with less than $v$ delay transitions. Note that $s_{\alpha_{v-1}}$ is the same state for every $h \in [H]^{d_1,...,d_v}$. Thus, we have $[H]^{d_1,...,d_v} = \{h \cdot h' \mid h \in [H_{\leq \alpha_{v-1}}]^{d_1,...,d_{v-1}} \textup{ and } h' \in [H_{\geq \alpha_{v-1}}]^{d_v}\}$. That is, the set of histories in $[H]^{d_1,...,d_v}$ is obtained by gluing together every prefix up to $s_{\alpha_{v-1}}$ with every suffix starting in $s_{\alpha_{v-1}}$. For two histories $h^m$ and $h^n$ in these two sets we denote their concatenation (where final state of $h^m$ is merged with initial state of $h^n$) by $h^{mn}$. This gives us 

$p^{d_1,...,d_v} = \displaystyle \sum_{h^m \in [H_{\leq \alpha_{v-1}}]^{d_1,...,d_{v-1}}} \sum_{h^n \in [H_{\geq \alpha_{v-1}}]^{d_v}} \bigg ( $
\\

$\displaystyle \prod_{i: \beta_i < \alpha_{v-1}} \delta(h^m_{\leq \beta_{i} - 1}, c^m_{\beta_i - 1})(a^m_{\beta_i})$
\\

$\displaystyle \prod_{0 \leq i < \alpha_{v-1}} \prod_{j \in \syn(a^m_i)} \sigma_j(\pi_j(h^m_{\leq i}))(c^m_i(j))$
\\

$\displaystyle \prod_{i: \beta_i > \alpha_{v-1}} \delta(h^{mn}_{\leq \beta_{i} - 1}, c^{mn}_{\beta_i - 1})(a^{mn}_{\beta_i})$
\\

$\displaystyle \prod_{\alpha_{v-1} \leq i < \alpha_v} \prod_{j \in \syn(a^{mn}_i)} \sigma_j(\pi_j(h^{mn}_{\leq i}))(c^{mn}_i(j)) \bigg)$
\\

$= \displaystyle \sum_{h^m \in [H_{\leq \alpha_{v-1}}]^{d_1,...,d_{v-1}}} $
\\

$\cdot \displaystyle \prod_{i: \beta_i < \alpha_{v-1}} \delta(h^m_{\leq \beta_{i} - 1}, c^m_{\beta_i - 1})(a^m_{\beta_i})$
\\

$\cdot \displaystyle \prod_{0 \leq i < \alpha_{v-1}} \prod_{j \in \syn(a^m_i)} \sigma_j(\pi_j(h^m_{\leq i}))(c^m_i(j))$
\\

$\cdot \displaystyle \sum_{h^n \in [H_{\geq \alpha_{v-1}}]^{d_v}} \bigg ( \prod_{i: \beta_i > \alpha_{v-1}} \delta(h^{mn}_{\leq \beta_{i} - 1}, c^{mn}_{\beta_i - 1})(a^{mn}_{\beta_i})$
\\

$\displaystyle \prod_{\alpha_{v-1} \leq i < \alpha_v} \prod_{j \in \syn(a^{mn}_i)} \sigma_j(\pi_j(h^{mn}_{\leq i}))(c^{mn}_i(j)) \bigg)$
\\

If we can show that the second sum is either $0$ or $1$ independently of the scheduler then we can apply the induction hypothesis on the remaining part. Note that no player can distinguish between the prefixes $h^m$ since they are in the same equivalence class. Thus, using pure strategies the players can only base their decision on what happens after reaching $s^{mn}_{\alpha_{v-1}}$. Now, using the same technique as in the proof of Lemma \ref{lem:pureunique} the result follows.
\qed
\end{proof}

From Lemma \ref{lem:purestrats} we know that the probabilities of events in $\sigmafieldb$ are independent of the scheduler when the strategy profile is pure. Using this we can show that this is also the case for non-pure strategy profiles. The idea of the proof is similar to the proof of Kuhn's Theorem \footnote{H. W. Kuhn, ''Extensive games and the problem of information,'' Annals of Mathematics Studies, vol. 28, 1953} behavioural strategies are shown equivalent in perfect recall extensive-form games. This intuition is applied in the 0-duration subhistories.

\begin{reftheo}{theo:sched_power}
Let $E \in \sigmafieldb$, let $\vect{s_0} \in S$ be a state, let $\sigma = (\sigma_1,...,\sigma_n)$ be a strategy profile and $\delta,\delta'$ be two schedulers. Then

$$\pr_{\vect{s_0}}^{\sigma, \delta}(E) = \pr_{\vect{s_0}}^{\sigma, \delta'}(E)$$
\end{reftheo}

\begin{proof}
We show this by showing that for every time-abstract cylinder $\cylta(H)$ for an interval-timed history $H = \vect{s_0} \vect{c_0} \arrow{a_1,I_1}{} \vect{s_1} \vect{c_1}  \cdots \arrow{a_k,I_k}{} \vect{s_k}$ such that $a_k \in \agt$ if $k > 0$, every strategy profile $\sigma$ and every pair $\delta, \delta'$ of schedulers we have
 
$$\pr_{\vect{s_0}}^{\sigma, \delta}(\cylta(H)) = \pr_{\vect{s_0}}^{\sigma, \delta'}(\cylta(H))$$

We use the same notation as in the proof of Lemma \ref{lem:purestrats}. Again we have
\\

$\displaystyle \pr_{\vect{s_0}}^{\sigma, \delta}(\cylta(H))$\\

$\displaystyle = \int_{\ell_1}^{u_1}... \int_{\ell_v}^{u_v} \prod_{1 \leq i \leq v} Q(s^m_{\alpha_i}, s^m_{\alpha_i + 1}) \cdot e^{-E(d_i)}$\\

$\displaystyle\cdot \sum_{h^m \in [H]^{d_1,...,d_v}} \bigg ( \prod_{1 \leq i \leq u} \delta(h^m_{\leq \beta_{i} - 1}, c^m_{\beta_i - 1})(a^m_{\beta_i})$\\

$\displaystyle\cdot \prod_{0 \leq i < k} \prod_{j \in \syn(a^m_i)} \sigma_j(\pi_j(h^m_{\leq i}))(c^m_i(j)) \bigg)  \; \mathrm{d} d_v \cdots \mathrm{d} d_1$ 
\\

For fixed delays $d_1,...,d_v$ we will show that a discrete probability distribution over a finite set of pure strategies gives rise to the same probability as above. As this is independent of the scheduler by Lemma \ref{lem:purestrats} so is the probability for the mixed strategies. As this holds for all delays, the Theorem follows.

Now, for each history $h^m \in [H]^{d_1,...,d_v}$ we define a pure strategy profile $\sigma^{h^m}$ that plays according to $h^m$ as well as a probability $p^{h^m}$ defined by
\\

$$p^{h^m} = \prod_{0 \leq i < k} \prod_{j \in \syn(a^m_i)} \sigma_j(\pi_j(h^m))(c^m_i(j))$$

If $p = \sum_{h^m \in [H]^{d_1,...,d_v}} p^{h^m} < 1$ then define a pure strategy $\sigma''$ that plays such that a $h \in [H]^{d_1,...,d_v}$ is not possible. Further, define $p^{\sigma''} = 1 - p$. Now, consider the experiment of using scheduler $\delta$ and picking either one of the strategy profiles $\sigma^{h^m}$ with probability $p^{h^m}$ or $\sigma''$ with probability $p^{\sigma''}$ and applying these strategies. The probability that a prefix of the play is in $[H]^{d_1,...,d_v}$ in this experiment is independent of the scheduler because of Lemma \ref{lem:purestrats}. We will show that it is in fact equal to $\pr_{\vect{s_0}}^{\sigma, \delta}(\cylta(H))$. Indeed, the probability is
\\

$\displaystyle \sum_{\sigma^{h^m}} p^{h^m} \sum_{h^m \in [H]^{d_1,...,d_v}}  \cdot \bigg ( \prod_{1 \leq i \leq v} Q(s^m_{\alpha_i}, s^m_{\alpha_i + 1}) \cdot e^{-E(d_i)}$\\

$\displaystyle\cdot \prod_{1 \leq i \leq u} \delta(h^m_{\leq \beta_{i} - 1}, c^m_{\beta_i - 1})(a^m_{\beta_i})$\\

$\displaystyle\cdot \prod_{0 \leq i < k} \prod_{j \in \syn(a^m_i)} \sigma^{h^m}_j(\pi_j(h^m_{\leq i}))(c^m_i(j)) \bigg)$ 
\\

$\displaystyle = \prod_{1 \leq i \leq v} Q(s^m_{\alpha_i}, s^m_{\alpha_i + 1}) \cdot e^{-E(d_i)}$\\

$\displaystyle\cdot \sum_{h^m \in [H]^{d_1,...,d_v}} \bigg ( p^{h^m} \cdot \prod_{1 \leq i \leq u} \delta(h^m_{\leq \beta_{i} - 1}, c^m_{\beta_i - 1})(a^m_{\beta_i}) \bigg )$
\\

By inserting the expression for $p^{h^m}$ the result follows.
\qed
\end{proof}

\section{Full proof of Proposition \ref{prop:reduc_decpomdp}}
\label{app:reduction}

\begin{refprop}{prop:reduc_decpomdp}
For a DEC-POMDP $\pomdp$ with $n$ players and target set $T$ we can construct in polynomial time a distributed IMC $\img$ with $n+1$ players and target set $T'$ such that 
 $$\exists \vect{\sigma}:  \pr^{\vect{\sigma}}_\img(\diamond T) = p 
 \;\; \iff \;\; 
 \exists \vect{\sigma}': \pr^{\vect{\sigma}'}_\pomdp(\diamond T') = p.
 $$
 
\end{refprop}

\begin{proof}
Let us fix a DEC-POMDP $\pomdp = (S, \agt, (\act_i)_{1 \leq i \leq n}, \delta, (\obs_i)_{1 \leq i \leq n}, \obsfunc, \initstate)$ with $n$ players $\agt = \{1,...,n\}$. 
Further, let $\act_i = \{a_{i1},...,a_{im_i} \}$ and $\obs_i = \{o_{i1},...,o_{i\ell_i} \}$ for player $i \in \agt$. 
The distributed IMC $\img$ has $n+1$ modules, one module for each player in $\pomdp$ and the \emph{main} module responsible for their synchronization. Intuitively,
\begin{itemize}
	\item the module of every player $i$ stores the last local observation in its state space. Every step of $\pomdp$ is modelled as follows: The player \emph{outputs} to the main module the action it chooses and then \emph{inputs} from the main module the next observation.
	\item The main module stores the global state in its state space. Every step of $\pomdp$ corresponds to the following: The main module \emph{inputs} the actions of all players one by one, then it randomly picks the new state and new observations according to the rules of $\pomdp$ based on the actions collected. The observations are lastly \emph{output} to all players, again one by one.
	The main module is constructed so in every state there is at most one action transition. Thus, there is only one trivial strategy that cannot influence anything.
\end{itemize}

The construction of module for player $i \in \{1,...,n\}$ is illustrated in Figure~\ref{fig:syncmodule_app} along with constructions for input and output. The interesting part is the inputting mechanism. Inputting an action from the set $\{a_1,\ldots,a_r\}$ when in a state $s$  is done as follows. Instead of waiting in $s$, the player travels by delay transitions in a round-robin fashion through a cycle of $r$ states, where in the $i$-th state, only the action $a_i$ is available. This way, the player cannot influence anything and must input the action that comes with probability 1.

\mycut{
Formally, for any player $i \in \agt$, we build a module $(S_i, \act'_i, \inter{}_i, \marko{}_i, \initstate_{i})$ where
\begin{itemize}
	\item $\states_i = \obs_i \cup \{?\}$,
	\item $\act'_i = \act_i \cup \obs_i$,
	
\end{itemize}
}

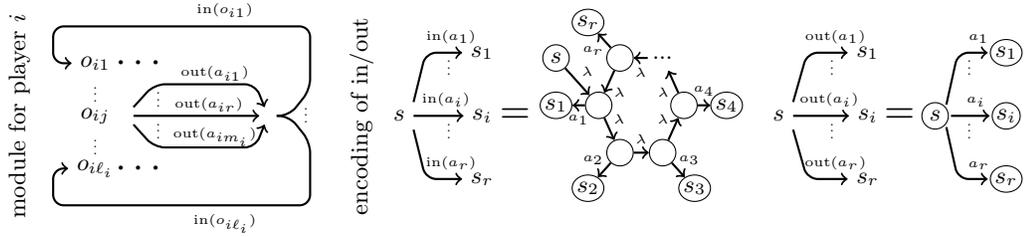
\begin{figure}
\begin{center}
\begin{tikzpicture}[xscale=0.7,yscale=0.7]

\tikzstyle{state}=[ellipse, draw=black,
inner sep=0pt, outer sep=1pt, minimum width=10pt, minimum height=10pt]
\tikzstyle{coord}=[inner sep = 0, outer sep=0]
\tikzstyle{dots}=[yscale=0.6, xscale=0.6]
\tikzstyle{bigdots}=[yscale=1.6, xscale=1.6]
\tikzstyle{action}=[font=\tiny,->,thick, rounded corners]
\tikzstyle{rate}=[font=\tiny,text width=2cm,text centered]
\tikzstyle{brace}=[decoration={brace, mirror},
decorate,font=\footnotesize]

\begin{scope}[xshift=0cm,yshift=-6cm,xscale=0.8]

\node[rotate=90] at (-2.5,0) {module for player $i$};

\begin{scope}[xshift=-0.7cm]
\node[] (s1) at (0,1) {$o_{i1}$};
\node[bigdots] at (1.1,1) {$\cdots$};
\node[dots] at (0,0.5) {$\vdots$};
\node[] (si) at (0,0) {$o_{ij}$};
\node[dots] at (0,-.5) {$\vdots$};
\node[] (sn) at (0,-1) {$o_{i \ell_i}$};
\node[bigdots] at (1.1,-1) {$\cdots$};
\end{scope}

\begin{scope}[xshift=0cm]
\node[] (si) at (0,0) {};
\node[] (i1) at (3.4,0) {};

	\path[action] (si) edge node[above=-2,pos=0.6] {out($a_{i r}$)} (i1);

	\draw[action] (si.north east) -- (0.7,0.6) -- node[above=-2,pos=0.6] {
				out($a_{i 1}$)
	} (3.1,0.6) -- (i1) ;
	
	\draw[action] (si.south east) -- (0.7,-0.6) -- node[above=-3,pos=0.6] {
				out($a_{i m_i}$)
	} (3.1,-0.6) -- (i1) ;
	
\node[dots] at (0.8,-.2) {$\vdots$};
\node[dots] at (0.8,.4) {$\vdots$};
\end{scope}

\begin{scope}[xshift=3.4cm]
\node[] (si) at (0,0) {};

\draw[action] (si.east) -- (0.5,0) -- (1,0.5) |- (0,1.7) -| node[above,pos=0.1] {
	in($o_{i 1}$)
} ($(s1)+(-1,0)$) -- (s1) ;

\draw[action] (si.east) -- (0.5,0) -- (1,-0.5) |- (0,-1.7) -| node[below,pos=0.1] {
	in($o_{i \ell_i}$)
} ($(sn)+(-1,0)$) -- (sn) ;

\node[dots] at (0.9,0.1) {$\vdots$};
\end{scope}

\end{scope}

\begin{scope}[xshift=5.4cm,yshift=-6cm,xscale=0.9]

\node[rotate=90] at (-1,0) {encoding of in/out};

\begin{scope}[xscale=0.7,xshift=-0.3cm]

\node (s) at (0,0) {$s$};

\node[] (si) at (2.5,0) {$s_i$};
\node[] (sr) at (2.5,-1.2) {$s_r$};
\node[] (s1) at (2.5,1.2) {$s_1$};

\path[action] (s) edge node[above,pos=0.65] {in($a_i$)} (si);
\draw[action] (s.north east) -- (0.8,1.2) -- node[above,pos=0.7] {in($a_1$)} (s1);
\node[dots] at (1.5,1) {$\vdots$};
\draw[action] (s.south east) -- (0.8,-1.2) -- node[above,pos=0.7] {in($a_r$)} (sr);
\node[dots] at (1.5,-.2) {$\vdots$};

\node[scale=1.6] at (3.5,0) {$=$};
\end{scope}

\begin{scope}[xshift=2.2cm,yshift=-5.2cm,xscale=0.9,yscale=0.9]
\tikzstyle{every node}=[ellipse, draw=black,
                        inner sep=0pt, minimum width=10pt, minimum height=10pt]

\draw (2,6)     node  (t1)	{};
\draw (2.5,5)     node  (t2)	{};
\draw (3.5,5)     node  (t3)	{};
\draw (4,6)     node (t4)	{};
\draw (3.5,7)     node [draw=none] (t5)	{...};
\draw (2.5,7)     node (t6)	{};

\draw (1,7)     node (s)	{$s$};
\draw (1,6)     node (s1)	{$s_1$};
\draw (1.75,4.25)     node (s2)	{$s_2$};
\draw (4.25,4.25)     node  (s3)	{$s_3$};
\draw (5,6)     node  (s4)	{$s_4$};
%\draw (4.5,8)     node [draw=none] (s5)	{...};
\draw (1.75,7.75)     node (s6)	{$s_r$};

\path[->,action] (t1) edge node [above right, draw=none] {$\lambda$} (t2);
\path[->,action] (t2) edge node [above, draw=none] {$\lambda$} (t3);
\path[->,action] (t3) edge node [above left, draw=none] {$\lambda$} (t4);
\path[->,action] (t4) edge node [below left, draw=none] {$\lambda$} (t5);
\path[->,action] (t5) edge node [below, draw=none] {$\lambda$} (t6);
\path[->,action] (t6) edge node [below right, draw=none] {$\lambda$} (t1);
\path[->,action] (s) edge node [above right, draw=none] {$\lambda$} (t1);

\path[->,action] (t1) edge node [below, draw=none] {$a_1$} (s1);
\path[->,action] (t2) edge node [above left, draw=none] {$a_2$} (s2);
\path[->,action] (t3) edge node [above right, draw=none] {$a_3$} (s3);
\path[->,action] (t4) edge node [above, draw=none] {$a_4$} (s4);
\path[->,action] (t6) edge node [below left, draw=none] {$a_r$} (s6);

\end{scope}

% % % % % % % % % % % % OUT

\begin{scope}[xshift=7.8cm,xscale=0.75]

	\node (s) at (0,0) {$s$};
	
	\node[] (si) at (2.5,0) {$s_i$};
	\node[] (sr) at (2.5,-1.2) {$s_r$};
	\node[] (s1) at (2.5,1.2) {$s_1$};
	
	\path[action] (s) edge node[above,pos=0.65] {out($a_i$)} (si);
	\draw[action] (s.north east) -- (0.8,1.2) -- node[above,pos=0.7] {out($a_1$)} (s1);
	\node[dots] at (1.5,1) {$\vdots$};
	\draw[action] (s.south east) -- (0.8,-1.2) -- node[above,pos=0.7] {out($a_r$)} (sr);
	\node[dots] at (1.5,-.2) {$\vdots$};
	
	\node[scale=1.6] at (3.4,0) {$=$};
	
	\begin{scope}[xshift=4.4cm,xscale=0.8]
	
	\node[state] (s) at (0,0) {$s$};
	
	\node[state] (si) at (2.5,0) {$s_i$};
	\node[state] (sr) at (2.5,-1.2) {$s_r$};
	\node[state] (s1) at (2.5,1.2) {$s_1$};
	
	\path[action] (s) edge node[above,pos=0.65] {$a_i$} (si);
	\draw[action] (s.north east) -- (0.8,1.2) -- node[above,pos=0.7] {$a_1$} (s1);
	\node[dots] at (1.2,1) {$\vdots$};
	\draw[action] (s.south east) -- (0.8,-1.2) -- node[above,pos=0.7] {$a_r$} (sr);
	\node[dots] at (1.2,-.2) {$\vdots$};
	
	\end{scope}

\end{scope}

\end{scope}

\end{tikzpicture}
\end{center}
\caption{Module for player $i$ on the left. Input and output encoding to the right.}
\label{fig:syncmodule_app}

\end{figure}

In the same fashion, the main module is constructed such that the extra player $n+1$ has at most one possible choice in each state and thus, this player controls nothing except what he is forced to do by the structure of his module. Thus, he has no choice but to enforce the rules of $\pomdp$ in $\img$. In the main module the current state $s$ of the DEC-POMDP $\pomdp$ is remembered at all times. Player $n+1$ then goes through $n$ steps inputting an action for each player. These are saved in the state of the module as well. When all actions have been collected, a random choice according to $\delta$ is made. This determines the successor state $s'$. Afterwards, a random choice of observations for the different players is done according to $\obsfunc$. Then, these observations are outputted to the other players again in a sequence of $n$ steps. The play proceeds for an infinite number of rounds, thereby modelling the DEC-POMDP.

The initial global state in $\img$ is known to all players and each initial local state is left one by one as the players output their first action. In addition, $\vect{s_0}(n+1) = \initstate$ as the main module mimics $\pomdp$. The target set $T'$ in $\img$ is given by all global states such that the main module is in a state corresponding to a state in $T$.

Since none of the players can affect the timing in the distributed IMC created (they always need to choose actions immediately and are never allowed to re-decide) none of them can gain any information that was not already there in the DEC-POMDP. Therefore, there exists a strategy profile $\vect{\sigma}$ in $\pomdp$ with $\pr^{\vect{\sigma}}_{\pomdp}(\diamond T) = p$ if and only if there exists a strategy profile ${\vect{\sigma}}'$ in $\img$ with $\pr^{{\vect{\sigma}}'}_{\img}(\diamond T') = p$. Note that this is also the case since there is never more than one enabled action at a time in the main module. Thus, player $n+1$ does not have any influence.
\qed
\end{proof}

\section{Full proof of Theorem \ref{theo:decpomdp_undec}}
\label{app:undec_decpomdp}

\begin{refprop}{theo:decpomdp_undec}
  The qualitative existence problem for DEC-POMDPs with $n \geq 2$ players is undecidable.
 
\end{refprop}

\begin{proof}
  We do a reduction from the non-halting problem of a deterministic Turing machine $M = (Q, q_0, \Sigma, \Delta, B, F)$ that never writes the blank symbol $B$ and starts with a blank input tape. Further, it never moves to the left of the initial state. Here, $Q$ is the finite set of control states, $q_0$ is the initial state, $\Sigma$ is the tape alphabet, $\Delta: Q \times \Sigma \rightarrow Q \times \Sigma \times \{L,R\}$ is the transition function and $F$ is the set of accepting states. From $M$ we construct a DEC-POMDP $\pomdp = (S, \{1,2\}, (\act_1, \act_2), \delta, (\obs_1, \obs_2), \obsfunc)$ with two players $\agt = \{1, 2\}$ such that $M$ does not halt if and only if player $1$ and $2$ have strategies $\sigma_1$ and $\sigma_2$ such that the probability of reaching a target set $T \subseteq S$ is 1. The overall structure of $\pomdp$ can be seen in Figure \ref{fig:undec_overall_app} but without the details of the sub-module $\pomdp'$.
 
\usetikzlibrary{shapes,backgrounds,calc}

\makeatletter
\tikzset{circle split part fill/.style  args={#1,#2}{%
 alias=tmp@name, % Jake's idea !!
  postaction={%
    insert path={
     \pgfextra{% 
     \pgfpointdiff{\pgfpointanchor{\pgf@node@name}{center}}%
                  {\pgfpointanchor{\pgf@node@name}{east}}%            
     \pgfmathsetmacro\insiderad{\pgf@x}
      %\begin{scope}[on background layer]
      %\fill[#1] (\pgf@node@name.base) ([xshift=-\pgflinewidth]\pgf@node@name.east) arc
      %                    (0:180:\insiderad-0.5\pgflinewidth)--cycle;
      %\fill[#2] (\pgf@node@name.base) ([xshift=\pgflinewidth]\pgf@node@name.west)  arc
      %                     (180:360:\insiderad-0.5\pgflinewidth)--cycle;  
      \fill[#1] (\pgf@node@name.base) ([xshift=-\pgflinewidth]\pgf@node@name.east) arc
                          (0:180:\insiderad-\pgflinewidth)--cycle;
      \fill[#2] (\pgf@node@name.base) ([xshift=\pgflinewidth]\pgf@node@name.west)  arc
                           (180:360:\insiderad-\pgflinewidth)--cycle;            %  \end{scope}   
         }}}}}  
 \makeatother  
 
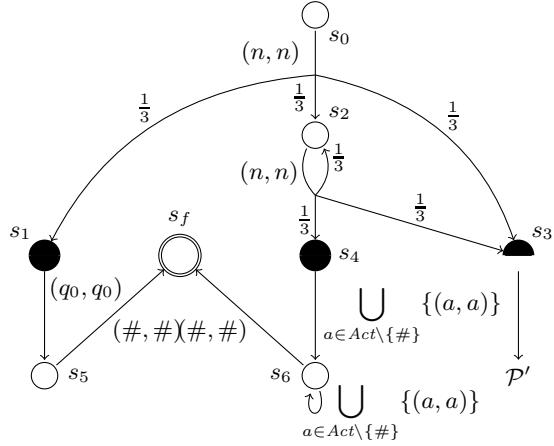
\begin{figure}
\begin{center}
\begin{tikzpicture}[xscale=0.9,yscale=0.8]

\tikzstyle{every node}=[ellipse, draw=black,
                        inner sep=0pt, minimum width=10pt, minimum height=10pt, outer sep=1pt]

\draw (4,7)     node [label=below right:$s_0$] (s0)	{};
\draw (4,6)     node [draw=none] (inv)	{};
\draw (0,3)     node[shape=circle split,
    draw=black,
    line width=0mm,text=white,font=\bfseries,
    circle split part fill={black,black}, label=above left:$s_1$] (s1)	{};
\draw (4,5)     node [label=above right:$s_2$] (s2)	{};
\draw (7,3)     node [shape=circle split,
    draw=black,
    line width=0mm,text=white,font=\bfseries,
    circle split part fill={black,white},label=above right:$s_3$] (s3)	{};
\draw (4,3)     node [shape=circle split,
    draw=black,
    line width=0mm,text=white,font=\bfseries,
    circle split part fill={black,black},label=right:$s_4$] (s4)	{};
\draw (2,3)     node [accepting, minimum width=15pt, minimum height=15pt,label=above:$s_f$] (sf)	{};

%\draw (2,2)     node [draw=none] (m1)	{$M_1$};
%\draw (4,2)     node [draw=none] (m2)	{$M_2$};
\draw (7,1)     node [draw=none] (m3)	{$\pomdp'$};

\draw (0,1)     node [label=right:$s_5$] (s5)	{};
\draw (4,1)     node [label=left:$s_6$] (s6)	{};

\path[-] (s0) edge node [left, draw=none,font=\footnotesize] {$(n,n)$} (4,6);
\path[->] (4,6) edge [bend right] node [above left, draw=none] {$\frac{1}{3}$} (s1);
\path[->] (4,6) edge node [left, draw=none] {$\frac{1}{3}$} (s2);
\path[->] (4,6) edge [bend left] node [above right, draw=none] {$\frac{1}{3}$} (s3);

\path[-] (s2) edge [bend right] node [left=-4, draw=none,font=\footnotesize] {$(n,n)$} (4,4);
\path[->] (4,4) edge [bend right] node [above right, draw=none] {$\frac{1}{3}$} (s2);
\path[->] (4,4) edge node [above right, draw=none] {$\frac{1}{3}$} (s3);
\path[->] (4,4) edge node [left=-2,pos=0.55, draw=none] {$\frac{1}{3}$} (s4);
%\path[->] (4,4) edge node [above left, draw=none] {$\frac{1}{4}$} (sf);

\path[->] (s3) edge (m3);
%\path[->] (s1) edge (m1);
%\path[->] (s4) edge (m2);
\path[->] (s1) edge node [pos=0.2,right=-5, draw=none,font=\footnotesize] {$(q_0,q_0)$} (s5);
\path[->] (s5) edge node [below right, draw=none,font=\footnotesize] {$(\#,\#)$} (sf);

\path[->] (s4) edge node [right = -12, draw=none,font=\footnotesize] {$\displaystyle\bigcup_{a \in \act \setminus \{\#\}}\{(a,a)\}$} (s6);
\path[->] (s6) edge node [below left, draw=none,font=\footnotesize] {$(\#,\#)$} (sf);
\path[->] (s6) edge [loop below] node [right=-20, draw=none,font=\footnotesize] {$\displaystyle\bigcup_{a \in \act \setminus \{\#\}}\{(a,a)\}$} (s6);

\end{tikzpicture}
\end{center}
\caption{Overall structure of $\pomdp$ without the details of $\pomdp'$.}
\label{fig:undec_overall_app}
\end{figure}

In this game, both players have two possible observations. In all states except $s_1,s_3,s_4$ player 1 receives one observation (illustrated with top-half of states filled with black) and in all other states he will receive another observation (top-half of states being white). Player 2 gets one observation for states $s_1,s_4$ as well as in some states in $\pomdp'$ (when bottom-half is black). In all other states he receives another observation.

The play starts in $s_0$ and with probability 1 either each player will receive the black observation exactly once during the game or the play will go to $s_f \in F$ from $s_2$ at some point. If the play goes to $s_1$ or $s_4$ the players will receive the observation at the same time and if the play goes to $s_3$ then player 2 will receive the observation in the next state. We will show that in order to have strategies $\sigma_1, \sigma_2$ that can make sure to reach the target set with probability 1 two things must be satisfied

\begin{enumerate}
 \item $\sigma_i$ must prescribe playing the $j$th configuration of the Turing machine $M$ (in a format described below) if player $i \in \{1,2\}$ receives the black observation in the $j$th step after the play leaves $s_0$
 
 \item $M$ does not halt
\end{enumerate}

Further, when these two properties are satisfied, players 1 and 2 can indeed reach the target set with probability 1 by applying $\sigma_1$ and $\sigma_2$. We now explain what we mean by playing a configuration of $M$. Let $(q,w,j) \in Q \times \Sigma^* \times \mathbb{N}$ be a configuration of $M$ where $q$ is the current control state, $w$ is the current non-blank part of the tape contents and $j$ is the current position of the tape head. By playing configuration $(q,w,j)$ we mean playing the following sequence of actions (where possible actions are corresponding to control states, tape symbols and $\#$ as an end of non-blank tape marker)
$$w_1 w_2 ... w_{j-1} q w_j w_{j+1} ... w_{|w|} \#$$
In other words, the contents of the tape is played one symbol at a time and the control state is prior to contents of the tape cell that the tape head points to.

The DEC-POMDP $\pomdp$ is constructed such that

\begin{enumerate}
 \item If the play goes to $s_1$ both players must play the first configuration of $M$
 
 \item If the play goes to $s_4$ both players must play the same sequence of symbols and ending with a $\#$
 
 \item If the play goes to $s_3$ the two players must play configurations $C_1$ and $C_2$ respectively such that $C_2$ is a successor configuration of $C_1$ according to the transition function of $M$. Further, this must be done with an offset of 1 step because player 2 receives the black observation one step later than player 1. Finally, the play will go to a sink state if they play a halting state as part of the configurations, making it impossible to reach $T$ afterwards. This can be done with a finite module~$\pomdp'$~\cite{BK10}.
 
\end{enumerate}

By induction we can show that if player 1 and 2 have strategies $\sigma_1$ and $\sigma_2$ to ensure reaching $T$ with probability 1 then they must play the $j$th configuration of $M$ if receiving the black observation in the $j$th step after the play leaves $s_0$. Indeed, if one of the players gets the black observation for $j = 1$ then $s_1$ is a possible state. Thus, both players must play the first configuration when $j = 1$. Now, suppose as induction hypothesis it is true for some $j$. Now, if player 2 receives the black observation in the $(j+1)$th step then the play might have passed through $s_3$ in which case player 1 will have received the black observation in the $j$th step. Thus, to make sure to reach $T$ in this case player 2 must player a successor configuration of what player 1 plays. Thus, player 2 must play the $(j+1)$th configuration of $M$ when getting the black observation in the $(j+1)$th step. Now, if player 1 gets the black observation after $j+1$ steps, then it is possible that the play is in $s_4$. If this is the case then player 2 will also have received the black observation in the $(j+1)$th step. Since player 2 will then play the $(j+1)$th configuration of $M$ then player 1 has to do this as well in order to reach $T$ since the players must play the same sequence of symbols to reach $T$ from $s_4$. This concludes the induction step.

We have now shown that if player 1 and 2 have strategies $\sigma_1$ and $\sigma_2$ to ensure reaching $T$ with probability 1 then they must play the $j$th configuration of $M$ when receiving the black observation in the $j$th step after the play leaves $s_0$. We now need to show that if $M$ does not halt, then applying these strategies will actually ensure reaching $T$ with probability 1 and if $M$ does halt, then applying these strategies will not ensure reaching $T$ with probability 1 (and thus, in this case, no strategies can ensure this).

In the case where $M$ does not halt, the play will reach $s_f, s_1, s_3$ or $s_4$ with probability 1. Suppose it reaches $s_1,s_3$ or $s_4$ at step $j$. Then both players play the appropriate configuration and reach $T$ with probability 1 since none of them will play a halting state at any point.

Suppose on the other hand that $M$ halts after $j$ steps. Now the play will reach $s_4$ after $j$ steps with positive probability. And since the players using $\sigma_1$ and $\sigma_2$ play the $j$th configuration of $M$ when this happens, they will play a halting state with positive probability. Thus, they cannot ensure reaching $T$ with probability 1 in this case.

In total this means the players have strategies to reach $T$ with probability 1 if and only if $M$ does not halt.
\qed
\end{proof}

\section{Proof of Proposition \ref{prop:correct}}

\begin{refprop}{prop:correct}
		The value of $\structure$ equals to the value of $\mdp_\structure$, i.e.
		$$
		\sup_{\vect{\sigma}} \pr^{\vect{\sigma}}(\diamond T) 
		\;\; = \;\; 
		\sup_\pi \pr^\pi(\diamond T).
		$$
\end{refprop}

The proof goes by three technical steps. First we define a stronger class of \emph{synchronization} strategies that can observe the whole global state whenever in a synchronization state that have the same power as the strategies in $\mdp_\structure$. Second, we show that any strategy profile can be mimicked by a synchronization strategy. The crucial step is the third. We show how any synchronization strategy can be (up to an arbitrary error) emulated by standard strategies in a distributed IMC. These strategies obtain (with high probability) full-observation by delaying.

\subsection{Synchronization strategies}
A \emph{global strategy} is a measurable function $\theta: \hist \to \Delta(\choice)$ that takes new choices only when allowed, i.e. that for any global history $h$ of the form $h = h' \cdot \vect{c} \arrow{a,t}{} \vect{s}$ and any global choice $\vect{c'}$ such that $\theta(h)(\vect{c'}) > 0$ we have $\vect{c}(j) = \vect{c'}(j)$ if $j\not\in\syn(a)$. For this section, we call standard strategies \emph{local} to stress the difference.

Furthermore, we call a global strategy a \emph{synchronization} strategy if it, intuitively, disregards non-local knowledge after the last synchronization state. Formally, the following condition needs to be satisfied for any pair of global histories $h$ and $h'$ and any player $j \in \{1,2\}$.
Let $h = h_1 \cdot s \cdot h_2$ and $h' = h'_1 \cdot s' \cdot h'_2$ where $h_2$ and $h'_2$ are the sequences after the last visit of a synchronization state ($s$ and $s'$). If (1) $h_1 \cdot s = h'_1 \cdot s'$, and (2) $\pi_j(h_2) = \pi_j(h_2)$, we have $\theta(h)(c) = \theta(h')(c)$ for any $c \in \choice_j$ where 
$\theta(h)(c) := \sum_{\vect{c}, \vect{c}(j)=c} \theta(h)(\vect{c})$.

Note that a synchronization strategy replaces a complete profile of local strategies. The definition of the probability measure $\pr_{\vect{s_0}}^{\theta,\delta}$ induced by a synchronization strategy $\theta$ and a scheduler $\delta$ goes along the same lines as for a strategy profile. The only difference is that we replace the term $St_i$ by the much simpler term $St'_i = \theta(h_{i})(\vect{c_i})$. 
By observing the whole synchronization state the strategy $\theta$ can completely avoid that a scheduler $\delta$ has any power. Indeed, the strategy $\theta$ can make sure that $\enab(\vect{s})$ always contains at most one action. Similarly to local strategies, we thus simplify the notation by denoting the probability measure $\pr_{\vect{s_0}}^{\theta}$.

\begin{lemma}
	The value of $\structure$ w.r.t. synchronization strategies equals to the value of $\mdp_\structure$:
	$$
	\sup_\pi \pr^{\pi}(\diamond T) 
	\;\; = \;\; 
	\sup_\theta \pr^\theta(\diamond T)
	$$
\end{lemma}

\begin{proof}
	Ad $\leq$: We show that an optimal pure memoryless strategy $\pi$ induces a synchronization pure memoryless strategy $\theta$ with the same value. Let $h$ be a history in $\structure$. This gives us a history $\bar{h}$ in the MDP (by removing choices, actions, time, and all private states). Further, let $(\vect{c},\sigma_1,\sigma_2) = \pi(\bar{h})$ be the decision of $\pi$.
	\begin{itemize}
		\item If $\last(h)$ is a synchronization state, $\theta$ takes the choice $\vect{c}$ (maybe step by step as the choice of each player can be changes only after the player moves).
		\item If $\last(h)$ is a private state and the player $j \in \{1,2\}$ needs to take an action, we distinguish two cases:
		\begin{itemize}
			\item if $\last(h)(j)$ is a private state, player $j$ takes action according to the pure memoryless strategy $\sigma_j$;
			\item if $\last(h)(j)$ is a synchronization state, player $j$ chooses $\donothing_j$.
		\end{itemize}
	\end{itemize}
	The equality $\pr^{\pi}(\diamond T) = \pr{\theta}(\diamond T)$ can be easily shown by induction in the number of visits to synchronization states. The crucial fact is that the probability to reach in $\structure$ by $\theta$ a state $(s'_1,s'_2) \in S'$ from a state $(s_1,s_2) \in S'$ only via $S\setminus S'$ coincides with $P((s_1,s_2),\pi((s_1,s_2)))(s'_1,s'_2)$.
	
	Ad $\geq$: For a global history $h \cdot (s_1, s_2)$, we denote by $\pr^\theta_{h \cdot (s_1,s_2)}(\diamond T)$ the probability to reach $T$ with $\theta$ after history $h \cdot (s_1,s_2)$ has already passed.	We show the statement in three steps. 
	
	First we argue that 
	$$\sup_\theta \pr_{h\cdot (s_1,s_2)}^\theta(\diamond T) = \sup_\theta \pr_{h'\cdot (s_1,s_2)}^\theta(\diamond T)$$ for any histories $h \cdot(s_1,s_2)$ and $h' \cdot(s_1,s_2)$ ending in the same synchronization state. Let us assume that for some $\theta'$ we have that $\pr_{h\cdot (s_1,s_2)}^{\theta'}(\diamond T)$ is strictly greater than the value for $h' (s_1,s_2)$. Since only the last synchronization state has impact on the further evolution, we could define a strategy $\theta''$ to play in $h'\cdot (s_1,s_2)$ in the same way as $\theta'$ plays in $h\cdot (s_1,s_2)$. This way, we get the same reachability probability yielding a contradiction.
	
	Second, we can thus define the \emph{value} of a synchronization state $(s_1,s_2)$ simply as 
	$$val(s_1,s_2) =  \sup_\theta \pr_{(s_1,s_2)}^\theta(\diamond T)$$ and we can easily obtain equations
	\begin{align*}
	val(s_1,s_2) 
	&= \sup_\theta \sum_{(s'_1,s'_2) \in S'} \pr_{(s_1,s_2)}^\theta(R(s'_1,s'_2)) \cdot val(s'_1,s'_2).
	\intertext{where $R(s'_1,s'_2)$ are the plays that reach $(s'_1,s'_2)$ via states in $S \setminus S'$. Because a synchronization strategy needs to play in each such fragment exactly as a profile of local strategies, it is in turn equal to}
	&= \sup_{(\sigma_1,\sigma_2)} \sum_{(s'_1,s'_2) \in S'} \pr_{(s_1,s_2)}^{(\sigma_1,\sigma_2)}(R(s'_1,s'_2)) \cdot val(s'_1,s'_2)
	\intertext{which can be decomposed to two independent one-player games}
	&= \sup_{\sigma_1,\sigma_2} \sum_{(s'_1,s'_2) \in S'} \pr_{s_1}^{\sigma_1}(R(s'_1)) \cdot \pr_{s_2}^{\sigma_2}(R(s'_2)) \cdot val(s'_1,s'_2).
	\end{align*}
	Each such one player game is equivalent to a continuous-time Markov decision process where any achievable vector of (time-unbounded) reachability probabilities of terminal states can be equivalently achieved in the embedded discrete-time Markov chain. The set of achievable vectors of reachability probabilities forms a polytope with the corners given by pure memoryless strategies. One can easily see that the supremum above is always obtained for some corners of the two polytopes, i.e. for pure memoryless strategies.
	\qed
\end{proof}

\subsection{Proof of $\leq$:}

We show that for every profile ${\vect{\sigma}}$ of local strategies there is a synchronization strategy $\theta$ such that
	$$\sup_{\vect{\sigma}} \pr^{\vect{\sigma}}(\diamond T) \leq  \pr^\theta(\diamond T)$$

	The proof follows from the fact that for a profile ${\vect{\sigma}}$ of local strategies there is a synchronization strategy $\theta$ inducing equivalent probability measures.
	%for any scheduler $\delta$. 
	Indeed, for decisions of every player $i$ the synchronization strategy disregards the global knowledge and decides only based on the $i$th local projection of the history exactly as $\sigma_i$.

\subsection{Proof of $\geq$:}

Finally, we show that using the optimal strategy $\pi'$, we can define a sequence of local strategy profiles ${\vect{\sigma}}^1, {\vect{\sigma}}^2, \ldots$ such that 
$$
\lim_{i\to\infty} \pr^{{\vect{\sigma}}^i}(\diamond T)
= 
\pr^{\pi'}(\diamond T).$$

%Let $\val(s_1,s_2)$ denote the optimal probability of winning $\pr^{\pi'}_{(s_1,s_2)}(\diamond T)$ from state $(s_1,s_2)$. Further, 
%Further, let $(a_{(s_1,s_2)},\sigma_1$ be 

Let us fix $i \in \Nset$. In order to define the strategy profile ${\vect{\sigma}}^i$, we need some auxiliary notions. 
We assume w.l.o.g. that $S'_1 = \{0, \ldots, |S_1|-1\}$ and $S'_2 = \{0, \ldots, |S_2|-1\}$.
Let $\lambda$ be the minimal rate of the delay transition over all synchronization states of both players and $K_i$ be the minimal integer such that $e^{-\lambda K_i} \leq 1/(4i\cdot(|S'_1|+|S'_2|))$. 
We split for each count of synchronization steps $k \in \Nseto$ the time into \emph{phases} of length $|S'_1| \cdot 2|S'_2| \cdot (K_i + k/\lambda)$ 
each composed of $|S'_1| \cdot 2 |S'_2|$ \emph{slots} of length $\ell_{i,k} = (K_i+k/\lambda)$. For any $t \in \Rset$ we define functions
\begin{itemize}
	\item $\row_{i,k}(t) := \lfloor t \; / \; { 2|S'_2| \cdot \ell_{i,k}}\rfloor \mod{|S'_1|}$,
	\item $\col_{i,k}(t) := \lfloor t \; / \; \ell_{i,k} \rfloor \mod{2|S'_2|}$.
\end{itemize}
that give the current row and column in a synchronization table depicted in the figure on page \pageref{fig:grid}.
%Then we define 
%\begin{itemize}
%	\item $\one_{i,k}(t) := \row_{i,k}(t)$, 
%	\item $\two_{i,k}(t) := \lfloor \col_{i,k}(t) \; / \; 3 \rfloor$, and
%	\item $\sync_{i,k}(t) := true$ iff $\col_{i,k}(t) \mod{3} \in \{0,1\}$.
%\end{itemize}
% Note that $\one_{i,k}(t) \in S_1$ and $\two_{i,k}(t) \in S_2$.

\noindent
Let $h$ be a history that includes $k$ synchronization actions. 
\begin{itemize}
	\item Let $h$ end in a synchronization state $(s_1,s_2)$ at time $t$, and let the timed guesses be $\bar{s}_1 = \row_{i,k}(t)$ and $\bar{s}_2 = \lfloor \col_{i,k}(t) \; / \; 2 \rfloor$. Furthermore, let $(c_1^{\bar{s}_1,\bar{s}_2},c_2^{\bar{s}_1,\bar{s}_2})$ be the optimal choices of $\pi'$ for players $1$ and $2$ in state $(\bar{s}_1,\bar{s}_2)$. Finally, we set a proposition $\sync$ to true iff $\col_{i,k}(t) \mod{2} = 0 $. 
	We define for each player (if the player is allowed to change the choice at the moment)
	\begin{align*}
	\sigma^i_1(h) & = \begin{cases}
	c_1^{\bar{s}_1, \bar{s}_2} & \text{if $\bar{s}_1 = s_1$ and $\sync$,} \\
	\donothing_1 & \text{otherwise;}
	\end{cases} \\
	\sigma^i_2(h) & = \begin{cases}
	c_2^{\bar{s}_1, \bar{s}_2} & \text{if $\bar{s}_2 = s_2$ and $\sync$,} \\
	\donothing_2 & \text{otherwise;}
	\end{cases} 
	\end{align*}
	\item If $h$ ends in a private state $(s_1,s_2)$, let $t$ be the time of last synchronization. We define again the timed guesses by $\bar{s}_1 = \row_{i,k}(t)$ and $\bar{s}_2 = \lfloor \col_{i,k}(t) \; / \; 2 \rfloor$.
	Then $\sigma^i_j$ plays in $s_j$ as the pure memoryless strategy $\sigma_j$ chosen by the optimal strategy $\pi'$ in state $(\bar{s}_1,\bar{s}_2)$.
\end{itemize}

\begin{mylemma}
	For any $i \in \Nset$ we have
	$$
	|\pr_{s}^{\pi'}(\diamond T)
	-  \pr_{s}^{{\vect{\sigma}}^i}(\diamond T)| 
	\;\leq \; \frac{1}{i}.
	$$
\end{mylemma}

\begin{proof}
	Let us fix $i \in \Nset$ and $k\in\Nseto$. From the definition of $\ell_{i,k}$ we have that the probability $err_{i,k}$ that a choice is \emph{not} updated within the length of one slot is 
	$$err_{i,k} \leq e^{-\lambda(K_i+k/\lambda)} \leq e^{-\lambda K_i} / 2^k \leq 1 / (4i\cdot (|S_1| + |S_2|) \cdot 2^k).$$ 
	Thus, the probability $corr^1_{i,k}$ that the correct synchronization is achieved within one phase is at least $corr_{i,k} \geq 1 - 2 err_{i,k}$ and the probability $inco^1_{i,k}$ that an incorrect synchronization is achieved within one phase is at most $inco^1_{i,k} \leq err_{i,k} \cdot (|S_1| + |S_2|)$ (here we bound it by the event that any synchronization choice is not switched off within the next $\neg\sync$ slot). With the remaining probability $wait^1_{i,k} = 1 - corr^1_{i,k}- inco^1_{i,k}$, no synchronization is achieved within one phase.
	
	The overall probability that the $k$th synchronization is incorrect is
	\begin{align*}
	inco_{i,k} 
	&= inco^1_{i,k} \cdot \sum_{j=0}^\infty \left(wait^1_{i,k}\right)^j  \leq inco^1_{i,k} \cdot \sum_{j=0}^\infty \left(1- corr^1_{i,k}\right)^j \\
	& \leq err_{i,k} \cdot (|S_1| + |S_2|) \cdot \sum_{j=0}^\infty \left(2 err_{i,k} \right)^j 
	%\\
	%& = err_{i,k} \cdot (|S_1| + |S_2|) \cdot \frac{1}{1- 2 err_{i,k}} 
	\leq \frac{1}{2\cdot i \cdot 2^k}.
	\end{align*}
	Therefore, the overall probability that \emph{any} synchronization is incorrect is $\sum_{j=0}^\infty inco_{i,j} \leq 1/i$. It is easy to see that every profile ${\vect{\sigma}}^i$ emulates $\pi'$, provided every synchronization guarantees that both players have correct knowledge of the other player's state.
	\qed
\end{proof}